\documentclass[submission,copyright,creativecommons]{eptcs}
\providecommand{\event}{ACCAT 12} % Name of the event you are submitting to
\usepackage{breakurl}             % Not needed if you use pdflatex only.

\usepackage{amssymb}
\usepackage{amsthm}
\usepackage{amsmath}
\usepackage{subfigure}
\usepackage[T1]{fontenc}
\usepackage[utf8]{inputenc}
\usepackage{color}
\usepackage{subfigure}
\usepackage{hyperref}

\usepackage[section]{placeins} %For FloatBarrier

\usepackage{aliascnt}

\usepackage{textcomp}
\usepackage[xspace]{ellipsis}
\usepackage{ragged2e}
\usepackage[nodayofweek]{datetime} % for \timestamp macro
\usepackage{soul}

\usepackage{pgf,tikz}
\usetikzlibrary{trees,decorations,arrows,automata,shadows,positioning,plotmarks,backgrounds,shapes}
\usetikzlibrary{calc,matrix,fit,petri,decorations.pathmorphing,shapes}

\usepackage[all]{xy}

\title{Satisfaction, Restriction and Amalgamation of Constraints in the Framework of $\mathcal{M}$-Adhesive Categories}
\author{Hanna Sch\"olzel$^{1}$, Hartmut Ehrig$^{1}$, Maria Maximova$^{1}$, Karsten Gabriel$^{1}$, \\
and Frank Hermann$^{ 1,2 }$
\institute{
1) Institut f{\"{u}}r Softwaretechnik und Theoretische Informatik,
 Technische Universit\"at Berlin, Germany \\
2) University of Luxembourg, Interdisciplinary Centre for Security, Reliability and Trust}
\email{[hannas, ehrig, mascham, kgabriel, frank]@cs.tu-berlin.de}
}
\def\titlerunning{Satisfaction, Restriction and Amalgamation of Constraints}
\def\authorrunning{H. Sch\"olzel, H. Ehrig, M. Maximova, K. Gabriel \& F. Hermann}

\newcommand{\newaliastheorem}[4]{%
	\newaliascnt{#1}{#2}%
	\newtheorem{#1}[#1]{#3}%
	\aliascntresetthe{#1}%
	\expandafter\newcommand\csname #1autorefname\endcsname{#4}%
}

\newtheorem{theorem}{Theorem}[section]
\numberwithin{theorem}{section}
\newaliastheorem{definition}{theorem}{Definition}{Def.}
\newaliastheorem{lemma}{theorem}{Lemma}{Lem.}
\newaliastheorem{remark}{theorem}{Remark}{Rem.}
\newaliastheorem{example}{theorem}{Example}{Ex.}
\newaliastheorem{fact}{theorem}{Fact}{Fact}
\newaliastheorem{corollary}{theorem}{Corollary}{Cor.}
\newtheorem*{genAssumption}{General Assumption}
\DeclareMathAlphabet{\mathcal}{OMS}{cmsy}{m}{n}

\usepackage{ifthen}
\newcounter{linesTODO}
\newlength{\XInnerSep}
\newcommand{\TG}{\ensuremath{\mathit{TG}}}

\newcommand{\M}{\ensuremath{\mathcal{M}}\xspace}

\newcommand{\cat}[1]{\ensuremath{\mathbf{#1}}\xspace}

\newcommand{\rarr}{\ensuremath{\rightarrow}}
\newcommand{\iDash}{\ensuremath{\stackrel{I}{\vDash}}}

\newcounter{comCounter}[page]
\newcommand{\linefill}{	%
	\cleaders
	\hbox{$\smash{\mkern-2mu\mathord-\mkern-2mu}$}	%
	\hfill
	\vphantom{\lower1pt\hbox{$\rightarrow$}}	%
	}
\newcommand{\xmake}[1]{\mathrel{\lower1pt\hbox{$#1$}}} % single-arrow transition
 \newcommand\wrapfill{\par
  \ifx\parshape\WF@fudgeparshape
  \nobreak
  \vskip-\baselineskip
  \vskip\c@WF@wrappedlines\baselineskip
  \allowbreak
  \WFclear
  \fi
}
\newcommand{\flabel}[1]{\label{f:#1}}

\newcommand{\fig}[3]{ % {FIGNAME}{FACTOR}{CAPTION} }
    \begin{figure*}[htb]
    \centering
    \includegraphics*[scale=#2]{figures/#1}
    \caption{#3} \flabel{#1}
    \end{figure*}
}

\newcounter{lines}

\let\qedorig\qed
\renewcommand{\qed}{\hspace*{\fill} \qedorig}

\let\existsorig\exists
\renewcommand{\exists}{\ \existsorig\ }
\let\forallorig\forall
\renewcommand{\forall}{\ \forallorig\ }
\let\nexistsorig\nexists
\renewcommand{\nexists}{\ \nexistsorig\ }

\newcounter{current}

\newcommand{\xcmatrix}[2]{\begin{center} \parbox[c]{0.5cm}{\xymatrix#1{#2}} \end{center}}

\newcommand{\ACleft}[2][-4.5ex,0ex]{\save[]+<#1>*{#2~\triangleright}\restore}
\newcommand{\ACright}[2][4.5ex,0ex]{\save[]+<#1>*{\triangleleft~#2}\restore}
\newcommand{\ACdown}[2][0ex,-3.6ex]{\save[]+<#1>*\txt{$\vartriangle$\\ $#2$}\restore}
\newcommand{\ACup}[2][0ex,3.6ex]{\save[]+<#1>*\txt{$#2$\\ $\triangledown$}\restore}

\begin{document}
\maketitle

\begin{abstract}
Application conditions for rules and constraints for graphs are well-known in the theory of graph transformation and have been extended already to \M-adhesive transformation systems. According to the literature we distinguish between two kinds of satisfaction for constraints, called general and initial satisfaction of constraints, where initial satisfaction is defined for constraints over an initial object of the base category. Unfortunately, the standard definition of general satisfaction is not compatible with negation in contrast to initial satisfaction.

Based on the well-known restriction of objects along type morphisms, we study in this paper restriction and amalgamation of application conditions and constraints together with their solutions. In our main result, we show compatibility of initial satisfaction for positive constraints with restriction and amalgamation, while general satisfaction fails in general.

Our main result is based on the compatibility of composition via pushouts with restriction, which is ensured by the horizontal van Kampen property in addition to the vertical one that is generally satisfied in \M-adhesive categories.
%This is based on a result concerning compatibility of compositions via pushouts with restriction, where the proof requires the horizontal van Kampen property, in contrast to the vertical one required for \M-adhesive categories.
%\timestamp
\end{abstract}

% =====================================================
% uncomment lines below to remove all update markers
% =====================================================
%\renewcommand{\updateMM}[1]{#1}
%\renewcommand{\updateFH}[1]{#1}
% =====================================================

\section{Introduction}
\label{sec:intro}
The framework of $\M$-adhesive categories has been introduced recently~\cite{EGH10,BEGG10} as a generalization 
of different kinds of high level replacement systems based on the double pushout (DPO) approach~\cite{EEPT06}.
Prominent examples that fit into the framework of $\M$-adhesive categories
are (typed attributed) graphs~\cite{EEPT06,LBE+07} and (high-level) Petri nets~\cite{BEE+09,EHP+08}.
In the context of domain specific languages and model transformations based on graph transformation, graph conditions (constraints) are already used extensively for the specification of model constraints and the specification of application conditions of 
transformation rules. 
%Graph conditions can be nested, may contain Boolean expressions and are equivalent to first order logic on graphs\rot{~\cite{HP05,HP09,Courcelle97theexpression,Rensink04}}.
Graph conditions can be nested, may contain Boolean expressions~\cite{HP05,HP09} and are expressively equivalent to first-order formulas on graphs~\cite{Courcelle97theexpression} as shown in~\cite{HP09,Rensink04}.
We generally use the term ``nested condition'' whenever we refer to the most general case.

%Based on the well-known restriction of objects along type morphisms we study in this paper restriction and amalgamation of 
%application conditions and constraints together with their solutions. 
Restriction is a general concept for the definition of views of domain languages and
is used for reducing 
the complexity of a model and for increasing the focus to relevant model element types.
A major research challenge in this field is to provide general results that allow for reasoning on
properties of the full model (system) by analyzing restricted properties on the views (restrictions) of the model only.
Technically, a restriction of a model is given as a pullback along type morphisms.
While this construction can be extended directly to restrictions of nested conditions, the satisfaction 
of the restricted nested conditions is not generally guaranteed for the restricted models, but---as we show in this paper---can be ensured under some sufficient conditions.

According to the literature~\cite{HP09,EEPT06}, we distinguish between two kinds of satisfaction for nested conditions, called general and initial satisfaction, where initial satisfaction is defined for nested conditions over an initial object of the base category. 
Intuitively, general satisfaction requires that a property holds for all occurrences of a premise pattern, while initial satisfaction requires this property for at least one occurrence.
Unfortunately, the standard definition of general satisfaction is not compatible with the Boolean operators for negation and disjunction,
but initial satisfaction is compatible with all Boolean operators (see App.~A in~\cite{SEM+12}).
In order to show, in addition, compatibility of initial satisfaction with restriction, we introduce the concept of amalgamation
for typed objects, where objects can be amalgamated along their overlapping according to the given type restrictions.

As the main technical result, we show that solutions for nested conditions can be composed and decomposed along an amalgamation
of them (\autoref{thm:ACviaRestr}), if the nested conditions are positive, i.e., they contain neither a negation nor a ``for all'' expression (universal quantification). 
Based on this property, we show in our main result (\autoref{thm:compatibility-initial-satisfaction-amalgamation}),
that initial satisfaction of positive nested conditions is compatible with amalgamation based on restrictions that agree on their overlappings. Note in particular that this result does not hold for general satisfaction which we illustrate by a concrete counterexample.

The structure of the paper is as follows.
Section~\ref{sec:general} reviews the general framework of $\M$-adhesive categories and 
main concepts for nested conditions and their satisfaction. Thereafter, \autoref{sec:restriction} 
presents the restriction of objects and nested conditions along type object morphisms.
Section~\ref{sec:amalgamation} contains the constructions and results concerning the amalgamation of objects and nested conditions
and in \autoref{sec:compatibility}, we present our main result showing the compatibility of initial satisfaction with amalgamation and restriction. 
Related work is discussed in \autoref{sec:relatedWork}.
Section~\ref{sec:conclusion} concludes the paper and discusses aspects of future work.
Appendix~\ref{sec:appendixB} contains the proofs that are not contained in the main part.
Additionally, App.~A in~\cite{SEM+12} provides formal details concerning  the transformation between both satisfaction relations 
and, moreover, their compatibility resp. incompatibility with Boolean operators. 
%\rot{For formal details concerning  the transformation between both satisfaction relations and their compatibility resp. incompatibility with Boolean operators as well as for proofs that are not contained in this paper consult~\cite{SEM+12}. 
%%App. A and B in~\cite{SEM+12}, respectively.
%}

\section{General Framework and Concepts}
\label{sec:general}
In this section we recall some basic well-known concepts and notions and introduce some new notions that we are using in our approach. Our considerations are based on the framework of \M-adhesive categories. An $\mathcal{M}$-adhesive category \cite{EGH10} consists of a category $\cat{C}$ together with a class $\mathcal{M}$ of monomorphisms as defined in \autoref{def:MAdhCat} below. The concept of $\mathcal{M}$-adhesive categories generalizes that of adhesive \cite{LS04}, adhesive HLR \cite{EHPP06}, and weak adhesive HLR categories \cite{EEPT06}.
\begin{definition}[\M-Adhesive Category]\label{def:MAdhCat} 
An \M-adhesive category $(\cat{C},\M)$ is a category $\cat{C}$ together with a class $\M$ of monomorphisms satisfying:
	\begin{itemize} \itemsep0pt
		\item the class $\M$ is closed under isomorphisms, composition and decomposition,
		%\begin{itemize}
%				\item composition ($f, g\in \M \Rightarrow g \circ f \in \M$) and
%				\item decomposition ($g \circ f \in \M \wedge g\in \M \Rightarrow f \in \M$)
%			\end{itemize}
				%	\vspace{3 mm}
		\item $\cat{C}$ has pushouts and pullbacks along \M-morphisms,
		\item \M-morphisms are closed under pushouts and pullbacks, and
		\item it holds the vertical van Kampen (short VK) property. This means that pushouts along \M-morphisms are \M-VK squares, i.\,e., pushout $(1)$ with $m \in \M$ is an \M-VK square, if for all commutative cubes $(2)$ with $(1)$ in the bottom, all vertical morphisms $a, b, c, d \in \M$ and pullbacks in the back faces we have that the top face is a pushout if and only if the front faces are pullbacks.
		
\begin{center}

\vspace{-1ex}
		$
			\begin{array}[t]{c}
%			  \SelectTips{cm}{}
				\xymatrix@C-2ex@R-2ex{
				                & A \ar[dl] \ar[rd]^m \ar@{}[dd]|{\textstyle (1)} \\
				C \ar[rd]     &
				                & B \ar[ld] \\
				                & D
				}
			\end{array}
		\hspace{2cm}
			\begin{array}[t]{c}
%			\SelectTips{cm}{}
				\xymatrix@L+0ex@C-0ex@R-4.7ex{
				                & & A' \ar[dll] \ar[rd] \ar@{.>}[ddd]_(0.3){a}  |!{[dr];[ddl]}\hole \\
				C' \ar[rd] \ar[ddd]_c
				                & & & B' \ar[dll] \ar[ddd]^b \\
				                & D' \ar[ddd]_(.3)d &  &  & {(2)}\\
				                & & A  \ar@{}[ddl]|{(1)} \ar@{.>}[dll] |!{[ul];[dl]}\hole
				                       \ar@{.>}[rd]^m \\
				C  \ar[rd]  & & & B  \ar[dll] \\
				                & D
				}
			\end{array}
		$

%
%
%\begin{tikzpicture}[scale=.8,->,shorten >=1pt,auto,node distance=2cm,semithick]
%
%\node[] (A) at (0,0) {$A$};
%\node[] (1) at (1,-1) {$(1)$};
%\node[] (B) at (2,0) {$B$};
%\node[] (C) at (0,-2) {$C$};
%\node[] (D) at (2,-2) {$D$};
%\path[->] (A) edge node {$m$} (B);
%\path[->] (A) edge node[swap] {$$} (C);
%\path[->] (B) edge node {$$} (D);
%\path[->] (C) edge node[swap] {$$} (D);
%
%\begin{scope}[xshift=9cm]
%
%\node[] (A) at (0,-.5) {$A$};
%\node[] (B) at (2,-1) {$B$};
%\node[] (C) at (-3,-2) {$C$};
%\node[] (D) at (-1,-2.5) {$D$};
%
%\node[] (Ap) at (0,1.5) {$A'$};
%\node[] (Bp) at (2,1) {$B'$};
%\node[] (Cp) at (-3,0) {$C'$};
%\node[] (Dp) at (-1,-.5) {$D'$};
%
%\path[->] (A) edge[dashed] node[swap] {$m$} (B);
%\path[->] (A) edge[dashed] node[swap,near end] {$$} (C);
%\path[->] (B) edge node[near end] {$$} (D);
%\path[->] (C) edge node[swap,near start] {$$} (D);
%
%\path[->] (Ap) edge node[near end] {$$} (Bp);
%\path[->] (Ap) edge node[swap,near end] {$$} (Cp);
%\path[->] (Bp) edge node[near start] {$$} (Dp);
%\path[->] (Cp) edge node[near start] {$$} (Dp);
%
%\path[->] (Ap) edge[dashed] node[near start] {$a$} (A);
%\path[->] (Bp) edge node[near start] {$b$} (B);
%\path[->] (Cp) edge node[swap,near end] {$c$} (C);
%\path[->] (Dp) edge node[swap,near end] {$d$} (D);
%\end{scope}
%\end{tikzpicture}%
\end{center}	
\end{itemize}
\end{definition}

\begin{remark}\label{rem:HorizontalVKP} 
In \autoref{sec:restriction}, \autoref{sec:amalgamation} and \autoref{sec:compatibility} we will also need the horizontal VK property, where the VK property is only required for commutative cubes with all horizontal morphisms in \M (see \cite{EGH10}), to show the compatibility of object composition and the corresponding restrictions. Note moreover, that an \M-adhesive category which also satisfies the horizontal VK property is a weak adhesive HLR category \cite{EEPT06}. 
\end{remark}
A set of transformation rules over an \M-adhesive category according to the DPO approach constitutes an \M-adhesive transformation system \cite{EGH10}. For various examples (graphs, Petri nets, etc.) see \cite{EEPT06}.
%\begin{definition}[$\M$-Adhesive Transformation System]\label{def:MAdTrafoSys} 
%	Given an \M-adhesive category $(\cat{C}, \M)$. \\
%	An \M-adhesive transformation system $\ \mathit{AS} = (\cat{C},\M, P)$ has in addition a set $P$ of rules of the form $\rho = (L \stackrel{l}{\leftarrow} K \stackrel{r}{\rightarrow} R)$ with $l, r \in \M$.  
%\end{definition}

%We are now introducing a special kind of pushout called effective pushout. Intuitively, an effective pushout is constructed as a pushout over a given pullback. 
In \autoref{sec:restriction}, \autoref{sec:amalgamation} and \autoref{sec:compatibility} we are considering \M-adhesive categories with effective pushouts. According to \cite{LS05}, the formal definition is as follows.

\begin{definition}[Effective Pushout]\label{def:EfPO}\leavevmode\\
\begin{minipage}{10cm} \vspace{1mm}
Given \M-morphisms $a: B \to X$, $b: C \to X$ in an \M-adhesive category $(\cat{C},\M)$ and let $(A,p_1,p_2)$ be obtained by the pullback of $a$ and $b$. %the pullback of $a$ and $b$.
Then pushout $(1)$ of $p_1$ and $p_2$ is called effective, if the unique morphism $u: D \to X$ induced by pushout $(1)$ is an \M-morphism.
\end{minipage}
\begin{minipage}{5cm} \vspace{-7mm} \hspace{6mm}
$
\xymatrix@R-3.5ex@C+1ex{
& B \ar[dr]^{i_1} \ar@/^3ex/[drr]^(.7){a}
\\
A  \ar[ur]^{p_1} \ar[dr]_{p_2} \ar@{}[rr]|{(1)}
& & D \ar[r]^{u}
& X
\\
& C \ar[ur]_{i_2} \ar@/_3ex/[urr]_(.7){b}
}
$

%
%
%\begin{tikzpicture}[scale=.8,->,shorten >=1pt,auto,node distance=2cm,semithick]
%\node[] (A) at (0,0) {$A$};
%\node[] (1) at (2,-1) {$C$};
%\node[] (B) at (2,1) {$B$};
%\node[] (C) at (4,0) {$D$};
%\node[] (D) at (6,0) {$X$};
%\node[] (2) at (2,0) {$(1)$};
%\path[->] (A) edge node[scale=.8] {$p_1$} (B);
%\path[->] (A) edge node[swap,scale=.8] {$p_2$} (1);
%\path[->] (1) edge node[swap, scale=.8] {$i_2$} (C);
%\path[->] (B) edge node[scale=.8] {$i_1$} (C);
%\path[->] (C) edge node[scale=.8] {$u$} (D);
%\path[->] (B) edge[bend left=30] node[scale=.8] {$a$} (D);
%\path[->] (1) edge[bend right=30] node[swap,scale=.8] {$b$} (D);
%\end{tikzpicture}%	
\end{minipage}
\end{definition}

Nested conditions in this paper are defined as application conditions for rules in \cite{HP05}. Depending on the context in which a nested condition occurs, we use the terms application condition \cite{HP05} and constraint \cite{EEPT06}, respectively. Furthermore, we define positive nested conditions to be used in \autoref{sec:restriction}, \autoref{sec:amalgamation}, and \autoref{sec:compatibility} for our main results.

\begin{definition}[Nested Condition]\label{def:NestCond}
A \emph{nested condition} $ac_P$ over an object $P$ is inductively defined as follows: 
\begin{itemize}
\item $true$ is a nested condition over $P$.
\item For every morphism $a: P \rarr C$ and nested condition $ac_C$ over $C$, $\exists(a,ac_C)$ is a nested condition over $P$.
\item A nested condition can also be a Boolean formula over nested conditions. This means that also $\neg ac_P$, 
$\bigwedge_{i \in \mathcal{I}} ac_{P,i},$ and $\bigvee_{i \in \mathcal{I}} ac_{P,i}$ are nested conditions over $P$ for nested conditions $ac_P$, 
$ac_{P,i}$ $(i \in \mathcal{I})$ over $P$ for some index set $\mathcal{I}$.
\end{itemize}
Furthermore, we distinguish the following concepts:
\begin{itemize}
\item A nested condition is called \emph{application condition} in the context of rules and match morphisms.
\item A nested condition is called \emph{constraint} in the context of properties of objects.
\item A \emph{positive nested condition} is built up only by nested conditions of the form $true$, $\exists (a, ac)$, 
$\bigwedge_{i \in \mathcal{I}} ac_{P,i}\ $ and $\ \bigvee_{i \in \mathcal{I}} ac_{P,i}$, where $\ \mathcal{I} \neq \emptyset$.
\end{itemize}
\end{definition}

An example for a nested condition and its meaning is given below.

\begin{example}[Nested Condition]\label{exp:NC}
Given the nested condition $ac_P$ from \autoref{fig:Satisfaction} where all morphisms are inclusions.
%\fig{Constraint}{.8} {Example for nested condition \label{fig:Constraint}}
Condition $ac_P$ means that the source of every $\mathtt{b}$-edge has a $\mathtt{b}$-self-loop and must be followed by some $\mathtt{c}$-edge such that subsequently, there is a path in the reverse direction visiting the source and target of the first $\mathtt{b}$-edge with precisely one $\mathtt{c}$-edge and one $\mathtt{b}$-edge in an arbitrary order.
We denote this nested condition by $ac_P = \exists (a_1,true) \ \wedge \ \exists (a_2,\exists(a_3,true) \vee \exists(a_4, true))$.
\end{example}

We are now defining inductively whether a morphism satisfies a nested condition (see \cite{EEPT06}).

\begin{definition}[Satisfaction of Nested Condition]\label{def:SatOfAC}
Given a nested condition $ac_P$ over $P$, a morphism $p: P \to G$ \emph{satisfies} $ac_P$ (see \autoref{fig:SatAC}), written $p \vDash ac_P$, if: 
\begin{itemize}
\item $ac_P = true$, or
\item $ac_P = \exists (a,ac_C)$ with $a: P \to C$ and there exists a morphism $q:C \to G \in \M$ such that $q \circ a = p$ and $q \vDash ac_C$, or
\item $ac_P = \neg ac_P^\prime$ and $p \not\vDash ac_P^\prime$, or
\item $ac_P = \bigwedge_{i \in \mathcal{I}} ac_{P,i}$ and for all $i \in \mathcal{I}$ holds $p \vDash ac_{P,i}$ , or
\item $ac_P = \bigvee_{i \in \mathcal{I}} ac_{P,i}$ and for some $i \in \mathcal{I}$ holds $p \vDash ac_{P,i}$.
\end{itemize}
\end{definition}

%\begin{figure}	
%\centering
%\begin{tikzpicture}[scale=.8,->,shorten >=1pt,auto,node distance=2cm,semithick]
%\node[] (A) at (0,0) {$P$};
%\node[] (1) at (2,-0.75) {$=$};
%\node[] (B) at (4,0) {$C$\rlap{$\triangleleft ac_C$}};
%\node[] (C) at (2,-2) {$G$};
%%\node[] (D) at (3,-2) {$G$};
%\path[->] (A) edge node[scale=.8] {$a$} (B);
%%\path[->] ($(A.east)+1*(0cm,-.1cm)$) edge node[swap,scale=.8] {$post_1$} ($(B.west)+1*(0cm,-.1cm)$);
%\path[->] (A) edge node[swap, scale=.8] {$p$} (C);
%\path[->] (B) edge node[scale=.8] {$q \vDash ac_C$} (C);
%\end{tikzpicture}%	
%\caption{\label{fig:SatAC} Satisfaction of nested condition $ac_P$}
%\end{figure}

%\com{FH: combine Figs. 2+3 in one Fig. 2a + 2b}

In the following we distinguish two kinds of satisfaction relations for constraints: General \cite{EEPT06} and initial satisfaction \cite{HP09}. Initial satisfaction is defined for constraints over an initial object of the base category while general satisfaction is considered for constraints over arbitrary objects. Intuitively, while general satisfaction requires that a constraint $ac_P$ is satisfied by every \M-morphism $p: P \rightarrow G$, intial satisfaction requires just the existence of an \M-morphism $p:P \rightarrow G$ which satisfies $ac_P$.

%\begin{figure}[h]
%  \centering
%  \begin{minipage}[b]{5 cm}
%  	\begin{tikzpicture}[scale=.8,->,shorten >=1pt,auto,node distance=2cm,semithick]
%			\node[] (A) at (0,0) {$P$};
%			\node[] (1) at (2,-0.75) {$=$};
%			\node[] (B) at (4,0) {$C$\rlap{$\triangleleft ac_C$}};
%			\node[] (C) at (2,-2) {$G$};
%			\path[->] (A) edge node[scale=.8] {$a$} (B);
%			\path[->] (A) edge node[swap, scale=.8] {$p$} (C);
%			\path[->] (B) edge node[scale=.8] {$q \vDash ac_C$} (C);
%		\end{tikzpicture}%	
%    %\includegraphics{Dateiname 1} 
%    \caption{Satisfaction of nested condition $ac_P$}
%    \label{fig:SatAC}
%  \end{minipage}
%  \begin{minipage}[b]{5 cm}
%    %\includegraphics{Dateiname 2} 
%     \begin{tikzpicture}[scale=.8,->,shorten >=1pt,auto,node distance=2cm,semithick]
%			\node[] (A) at (0,0) {$I$};
%			\node[] (1) at (2,-0.75) {$=$};
%			\node[] (B) at (4,0) {$P$\rlap{$\triangleleft ac_P$}};
%			\node[] (C) at (2,-2) {$G$};
%			\path[->] (A) edge node[scale=.8] {$i_P$} (B);
%			\path[->] (A) edge node[swap, scale=.8] {$i_G$} (C);
%			\path[->] (B) edge node[scale=.8] {$p \vDash ac_P$} (C);
%		\end{tikzpicture}%
%    \caption{Initial satisfaction of constraint $ac_I$}
%    \label{fig:InitSatC}
%  \end{minipage}
%\end{figure}

\vspace*{-3ex}
\begin{figure}[h]
  \centering
  \subfigure[Satisfaction of $ac_P$ by morphism $p$]{
    \label{fig:SatAC}

    \quad \quad
$
\xymatrix@R-1.5ex@C+1ex{
P \ar[rr]^a \ar[dr]_p
& \ar@{}[d]|{=} & C \triangleleft ac_C \ar[dl]^{q \ \vDash \  ac_C}
\\
&G
}
$
%    \begin{tikzpicture}[scale=.8,->,shorten >=1pt,auto,node distance=2cm,semithick]
%			\node[] (A) at (0,0) {$P$};
%			\node[] (1) at (2,-0.75) {$=$};
%			\node[] (B) at (4,0) {$C$\rlap{$\triangleleft ac_C$}};
%			\node[] (C) at (2,-2) {$G$};
%			\path[->] (A) edge node[scale=.8] {$a$} (B);
%			\path[->] (A) edge node[swap, scale=.8] {$p$} (C);
%			\path[->] (B) edge node[scale=.8] {$q \vDash ac_C$} (C);
%		\end{tikzpicture}%	
		\quad \quad
  }
  \hspace*{10ex} %\hfill
  \subfigure[Initial satisfaction of $ac_I$]{
    \label{fig:InitSatC}
$
\xymatrix@R-1.5ex@C+1ex{
I \ar[rr]^{i_P} \ar[dr]_{i_G}
& \ar@{}[d]|{=} & P \triangleleft ac_P \ar[dl]^{p \ \vDash \  ac_P}
\\
&G
}
$
%    \begin{tikzpicture}[scale=.8,->,shorten >=1pt,auto,node distance=2cm,semithick]
%			\node[] (A) at (0,0) {$I$};
%			\node[] (1) at (2,-0.75) {$=$};
%			\node[] (B) at (4,0) {$P$\rlap{$\triangleleft ac_P$}};
%			\node[] (C) at (2,-2) {$G$};
%			\path[->] (A) edge node[scale=.8] {$i_P$} (B);
%			\path[->] (A) edge node[swap, scale=.8] {$i_G$} (C);
%			\path[->] (B) edge node[scale=.8] {$p \vDash ac_P$} (C);
%		\end{tikzpicture}%
  }
  \caption{Satisfaction of nested conditions}
  \label{fig:satisfaction}
\end{figure}
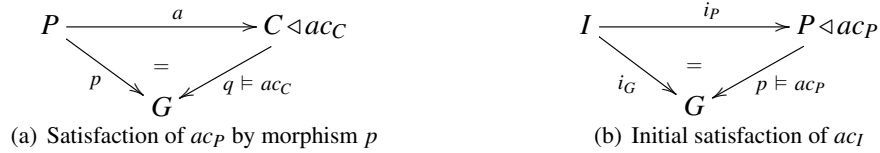

\begin{definition}[General Satisfaction of Constraints]\label{def:genSatisfaction}
Given a constraint $ac_P$ over $P$. An object $G$ generally satisfies $ac_P$, written $G \vDash ac_P$, if $\forall p: P \rightarrow G \in \M$. 
$p \vDash ac_P$ (see \autoref{fig:SatAC}).
\end{definition}

%\begin{definition}[Initial Satisfaction of Constraints]\label{def:initSatisfaction}
%Given a constraint $ac_I$ over initial object $I$. An object $G$ initially satisfies $ac_I$, written $G \stackrel{I}{\vDash} ac_I$, if $i_G \vDash ac_I$ for the initial morphism $i_G: I \rightarrow G$ (see Figure \ref{fig:InitSatC}). 
%\end{definition}

\begin{definition}[Initial Satisfaction of Constraints]\label{def:initSatisfaction}
Given a constraint $ac_I$ over an initial object $I$. An object $G$ initially satisfies $ac_I$, written $G \stackrel{I}{\vDash} ac_I$, if $i_G \vDash ac_I$ for the initial morphism $i_G: I \rightarrow G$.\\ Note, that for $ac_I = \exists(i_P, ac_P)$ we have 
\[ G \stackrel{I}{\vDash} ac_I \Leftrightarrow \exists p: P \rightarrow G \in \M.\ p \vDash ac_P\ \ \text{(see \autoref{fig:InitSatC})}. \] 
\end{definition}

This means that the general satisfaction corresponds more to the universal satisfaction of constraints while the initial satisfaction corresponds more to the existential satisfaction.
%\begin{figure}	
%\centering
%\begin{tikzpicture}[scale=.8,->,shorten >=1pt,auto,node distance=2cm,semithick]
%\node[] (A) at (0,0) {$I$};
%\node[] (1) at (2,-0.75) {$=$};
%\node[] (B) at (4,0) {$P$\rlap{$\triangleleft ac_P$}};
%\node[] (C) at (2,-2) {$G$};
%\path[->] (A) edge node[scale=.8] {$i_P$} (B);
%\path[->] (A) edge node[swap, scale=.8] {$i_G$} (C);
%\path[->] (B) edge node[scale=.8] {$p \vDash ac_P$} (C);
%\end{tikzpicture}%
%\caption{\label{fig:InitSatC} Initial satisfaction of constraint $ac_I$}
%\end{figure}

%\begin{remark}\label{rem:InitSatvsGenSat}
%For $ac_I = \exists(i_P, ac_P)$ we have: $G \stackrel{I}{\vDash} ac_I \Leftrightarrow \exists p: P \rightarrow G \in \M$. $p \vDash ac_P$. Moreover it holds: $G \vDash ac_P \Leftrightarrow G \stackrel{I}{\vDash} ac_P$ if $P$ is an \M-initial object, i.e., $I$ is an intial object and the initial morphisms $i_G: I \rightarrow G$ are in \M.
%\end{remark}

%\begin{remark}\label{rem:InitSatvsGenSat}
%$G \vDash ac_P \Leftrightarrow G \stackrel{I}{\vDash} ac_P$ if $P$ is an \M-initial object, i.e., $I$ is an intial object and the initial morphisms $i_G: I \rightarrow G$ are in \M.
%\end{remark}

For positive nested conditions, we define solutions for the satisfaction problem. A solution $Q$ (a tree of morphisms) determines which morphisms are used to fulfill the satisfaction condition.
%elimintating the need to determine solutions to the existential quantification as present in the former satisfaction definition (item 2).

\begin{definition}[Solution for Satisfaction of Positive Nested Conditions]\label{def:solution}
Given a positive nested condition $ac_P$ over $P$ and a morphism $p: P \to G$. Then $Q$ is a solution for $p \vDash ac_P$ if:
\begin{itemize}
	\item $ac_P = true$ and $Q = \emptyset$, or
	\item $ac_P = \exists(a, ac_C)$ with $a: P \to C$ and $Q = (q, Q_C)$ with \M-morphism $q: C \to G$ such that $q \circ a = p$ and $Q_C$ is
		a solution for $q \vDash ac_C$ (see \autoref{fig:SatAC}), or
	\item $ac_P = \bigwedge_{i \in \mathcal{I}} ac_{P,i}$ and $Q = (Q_i)_{i \in \mathcal{I}}$ such that $Q_i$ is a solution for $p \vDash ac_{P,i}$
		for all $i \in \mathcal{I}$, or
	\item $ac_P = \bigvee_{i \in \mathcal{I}} ac_{P,i}$ and $Q = (Q_i)_{i \in \mathcal{I}}$ such that there is $j \in \mathcal{I}$ with solution
	%for some $j \in \mathcal{I}$ there is
		$Q_j$ %a solution 
		for $p \vDash ac_{P,j}$ and for all $k \in \mathcal{I}$ with $k \neq j$ 
		it holds that $Q_k = \emptyset$.
\end{itemize}
\end{definition}

The following example demonstrates the general and initial satisfaction of constraints and gives their corresponding solutions.

\begin{example}[Satisfaction and Solution of Constraints]\label{exp:SatOfNestCond}\leavevmode
\begin{enumerate}
\item General Satisfaction\\
Consider the graph $G_A$ from \autoref{fig:Satisfaction} below and the constraint $ac_P$ from \autoref{exp:NC}. There are two possible \M-morphisms $p_1, p_2: P \rightarrow G_A$, where $p_1$ is an inclusion and $p_2$ maps $\mathtt{b_1}$ to $\mathtt{b_2}$ with the corresponding node mapping. For both matches $p_1$ and $p_2$, there is a $\mathtt{b}$-self-loop on the image of node $\mathtt{1}$, a $\mathtt{c}$-edge outgoing from the image of node $\mathtt{2}$, as well as the corresponding images for edges $\mathtt{b_2}$ and $\mathtt{c_2}$ in $C_3$. Thus, $G_A$ generally satisfies $ac_P$. \\
%The graph $G_A$ from the Figure \ref{fig:Satisfaction} below generally satisfies $ac_P$ from the Example \ref{exp:NC}.\\ 
%Let 
%%$a_1(i)=a_2(i)=i$ for $i \in \left\{1, 2\right\}$, $\ a_1(b_1)=a_2(b_1)=b_1$, $\ a_3(k)=a_4(k)=k$ for $k \in \left\{1, 2, 3\right\}$, and $\ a_3(j)=a_4(j)=j$ for $j \in \left\{b_1, c_1\right\}$. 
%$a_i(j)=j$ and $a_i(k)=k$ for $i \in \left\{1, 2, 3, 4\right\}, j \in \left\{1, 2, 3\right\}, k \in \left\{b_1, c_1\right\}$.
%The general satisfaction holds for $G_A$ since for every morphism $p: P \rightarrow G \in \mathcal{M}$ exist morphisms $q_i: C_i \rightarrow G_A \in \mathcal{M}$ for $i \in \left\{1, 2\right\}$ with $q_1(j)=q_2(j)=p(j)$ for $j \in \left\{1, 2\right\}$ and $q_1(b_1)=q_2(b_1)=p(b_1)$ such that $q_i \circ a_i=p$ for $i \in \left\{1, 2\right\}$ and $q_1 \vDash true$, $q_2 \vDash \exists(a_3,true) \ \vee \ \exists(a_4,true)$, because there exists the morphism $q_3: C_3 \rightarrow G_A \in \mathcal{M}$ with $q_3(j)=p(j)$ for $j \in \left\{1, 2\right\}$ and $q_3(b_1)=p(b_1)$ such that $q_3 \circ a_3 = q_2$.\\
A corresponding solution for $p_1 \vDash ac_P$ is given by $Q_{gen} = (Q_i)_{i \in \left\{1, 2\right\}}$ with $Q_1 = (q_1, \emptyset)$ and $Q_2 = (q_2, (Q_j)_{j \in \left\{3, 4\right\}})$, where $Q_3=(q_3, \emptyset)$, $Q_4=\emptyset$ and $q_i: C_i \rightarrow G_A$ for $i = 1, 2, 3$ are inclusions.
	
\fig{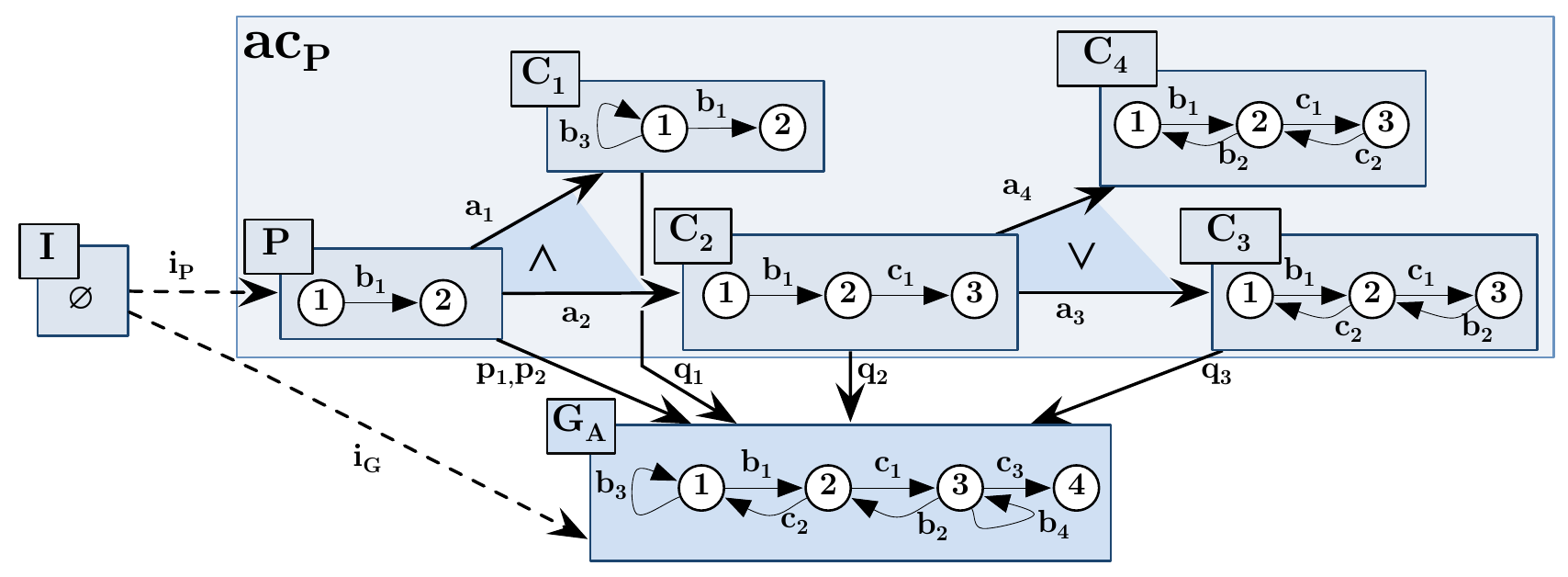}{.8} {General and initial satisfaction of constraints \label{fig:Satisfaction}}
	
\item Initial Satisfaction\\
Let 
%$ac_I = \exists (\emptyset \ \ \hookrightarrow \ \ \bigcirc_1 \stackrel{b_1}{\rightarrow} \bigcirc_2, \ ac_P)$, abbreviated as $\exists (i_P, ac_P)$
$ac_I = \exists (i_P, ac_P)$ with $i_P$ as depicted in \autoref{fig:Satisfaction} and $ac_P$ from \autoref{exp:NC}. 
The graph $G_A$ initially satisfies $ac_I$ since there is $p_1: P \rightarrow G_A \in \M$ satisfying $ac_P$ as mentioned before. \\
A corresponding solution for $i_G \vDash ac_I$ is given by $Q_{init} = (p_1, Q_{gen})$ with $Q_{gen}$ from the example for general satisfaction.
\end{enumerate}
\end{example}

\begin{remark}\label{rem:typedAC} 
A nested condition is called \emph{typed} over a given type object, if all nested conditions in every of its nesting levels are also typed over the same type object. Furthermore, matches and corresponding solutions are required to be compatible with this type of object as well.
%the compatibility of the corresponding match and solution with this type object is required.
\end{remark}

%\begin{definition}[Typed Nested Conditions and Their Solutions]\label{def:typedAC}\leavevmode\\
%	Given a nested condition $ac_P$ over $P$ and an object $TG$. We say that $ac_P$ is typed over $TG$ if $P$ is typed over $TG$ and
%	
%	\begin{minipage}{10cm} %\vspace{-8mm}
%	\begin{enumerate}
%		\item $ac_P = true$, or
%		\item $ac_P = \exists(a, ac_C)$ with $a: P \to C$ such that $ac_C$ is typed over $TG$ and $t_C \circ a = t_P$, or
%		\item $ac_P = \neg ac_P^\prime$ and $ac_P^\prime$ is typed over $TG$, or
%		\item $ac_P = \bigwedge_{i \in \mathcal{I}} ac_{P,i}$ or $ac_P = \bigvee_{i \in \mathcal{I}} ac_{P,i}$ and $ac_{P,i}$ is typed over $TG$
%			for all $i \in \mathcal{I}$.
%	\end{enumerate}
%\end{minipage}
%\begin{minipage}{5cm} \vspace{-8mm} \hspace{6mm}
%	\cmatrix{
%		P \ar[rr]^{a} \ar[ddr]_{t_P} \ar[dr]^{p} & & C \ACright{ac_C} \ar[ddl]^{t_C} \ar[dl]_{q} \\
%			& G \ar[d]|{t_G} \\
%			& TG &
%	}
%	\end{minipage}
%
%	Given a nested condition $ac_P$ typed over $TG$ and a typed object $t_G: G \to TG$, for a match $p: P \to G$ we require that
%	$p$ is compatible with the type morphisms, i.\,e.\ $t_G \circ p = t_P$. 
%	Moreover, for $ac_P = \exists(a, ac_C)$ and a solution $Q = (q, Q_C)$ we require that also $q$ is compatible with the type morphisms,
%	i.\,e.\ $t_G \circ q = t_C$.
%\end{definition}

%\section{Properties of General and Initial Satisfaction}
%\label{sec:properties}
%\input{2b-properties}

\section{Restriction Along Type Morphisms}
\label{sec:restriction}
In this section, we present the restriction of objects, morphisms, positive nested conditions and their solutions along type morphisms which are the basis for the amalgamation of nested conditions in \autoref{sec:amalgamation}. 

\begin{genAssumption}
In this and the following sections, we consider an \M-adhesive category $(\cat{C},\M)$ satisfying the horizontal VK property (see \autoref{rem:HorizontalVKP}) and has effective pushouts (see \autoref{def:EfPO}).
\end{genAssumption}

\begin{definition}[Restriction along Type Morphism] \label{def:restr} 
Given an object $G_A$ typed over $TG_A$ by $t_{G_A}\colon $ \linebreak ${G_A \rarr TG_A}$ and $t: TG_B \rightarrow TG_A \in \M$, then $TG_B$ is called restriction of $TG_A$, 
$G_B$ is a restriction of $G_A$, and $t_{G_B}$ is a restriction of 
$t_{G_A}$, if (1) is a pullback. Given $a: G'_A \rarr G_A$, then $b$ is a \emph{restriction} of $a$ along type morphism $t$, written $b = Restr_t(a)$, if (2) is a pullback. 
%Note that $Restr_t(t_{G_A}) = t_{G_B}$ if we build the pullback over $TG_B \stackrel{t}{\rarr} TG_A \stackrel{id}{\larr} TG_A$.
%\com{FH: Sentence seems to be wrong - remove? or replace $t_{G_B}$ by $t_{G_A}$}

\centering
$\begin{array}{l}%\\[1ex]
			\xymatrix{
			TG_A 
			& G_A \ar[l]_{t_{G_A}} 
			& G_A' \ar[l]_{a} \\
			TG_B \ar[u]^{t} 
			& G_B \ar[u]_{t_G} \ar[l]_{t_{G_B}} \ar@{}[ul]|{(1)}
			& G_B' \ar[u]_{t_G'} \ar[l]_{b} \ar@{}[ul]|{(2)}
			}
\end{array}$\\[1ex]
\end{definition}

For positive nested conditions, we can define the restriction recursively as restriction of their components.

\begin{definition}[Restriction of Positive Nested Conditions]\label{def:restrAC}
Given a positive nested condition $ac_{P_A}$ typed over $TG_A$ and let $TG_B$ be a restriction of it with $t: TG_B \rightarrow TG_A \in \M$. Then we define the restriction $ac_{P_B} = Restr_t(ac_{P_A})$ over the restriction $P_B$ of $P_A$ as follows: 
\begin{itemize}
\item The restriction of $true$ is $true$, 
\item the restriction of $\exists(a,ac_{C_A})$ is given by restriction of $a$ and $ac_{C_A}$, 
	i.\,e.,\ $ac_{P_B} = \exists(Restr_t(a),$ $Restr_t(ac_{C_A}))$, and 
\item the restriction of a Boolean formula is given by the restrictions of its components, i.\,e.,\  \linebreak
	$Restr_t(\neg ac_{P_A}^\prime) = \neg Restr_t(ac_{P_A}^\prime)$, 
	$Restr_t(\bigwedge_{i \in \mathcal{I}} ac_{P_A,i}) = \bigwedge_{i \in \mathcal{I}} Restr_t(ac_{P_A,i})$, and \linebreak
	$Restr_t(\bigvee_{i \in \mathcal{I}} ac_{P_A,i}) = \bigvee_{i \in \mathcal{I}} Restr_t(ac_{P_A,i})$.
\end{itemize}

\xcmatrix{@R-4ex}{
			& & P_A \ar@/_1.5ex/[dl]_{a} \\
			TG_A 
			& **[r] C_A \triangleleft ac_{C_A} \ar[l]
			& \\
			&& \\
			TG_B \ar[uu]^{t} 
			& **[r] C_B \triangleleft ac_{C_B}  \ar[uu]_{t_C} \ar[l]
			& \\
			& & P_B \ar@/^1.5ex/[ul]^{b} \ar[uuuu]_{t_P}
			}	

%	\xcmatrix{@R-1ex@C-1ex}{
%		& TG_A
%		& \\
%		P_A \ar[ur] \ar[rr]_/3ex/{a}
%		& 
%		& C_A \ACright{ac_{C_A}} \ar[ul] \\
%		& TG_B \ar[uu]^/2ex/{t}
%		& \\
%		P_B \ar[uu] \ar[ur] \ar[rr]^{b} 			&
%		& C_B \ACright{ac_{C_B}} \ar[uu] \ar[ul] 
%	}
\end{definition}

Now we extend the restriction construction to solutions of positive nested conditions and show in \autoref{fact:satisfAC} that a restriction of a solution is also a solution for the corresponding restricted constraint.

\begin{definition}[Restriction of Solutions for Positive Nested Conditions]\label{def:restrSolution}
	Given a positive nested condition $ac_{P_A}$ typed over $TG_A$ together with a restriction $ac_{P_B}$ along $t: TG_B \to TG_A$. 
	For a morphism $p_A: P_A \to G$ and a solution $Q_A$ for $p_A \vDash ac_{P_A}$, the restriction $Q_B$ of $Q_A$ along $t$,
	written $Q_B = Restr_t(Q_A)$, is defined inductively as follows:
	\begin{itemize}
		\item If $Q_A$ is empty then also $Q_B$ is empty,
		\item if $ac_{P_A} = \exists(a: P_A \to C_A, ac_{C_A})$ and $Q_A = (q_A, Q_{CA})$,
			then $Q_B = (q_B, Q_{CB})$ such that $q_B$ and $Q_{CB}$ are restrictions of $q_A$ respectively $Q_{CA}$, and
		\item if $ac_{P_A} = \bigwedge_{i \in \mathcal{I}} ac_{P_A,i}$ or $ac_{P_A} = \bigvee_{i \in \mathcal{I}} ac_{P_A,i}$,
			and $Q_A = (Q_{A,i})_{i \in \mathcal{I}}$, then
			$Q_B = (Q_{B,i})_{i \in \mathcal{I}}$ such that $Q_{B,i}$ is a restriction of $Q_{A,i}$ for all $i \in \mathcal{I}$.
	\end{itemize} 
\end{definition}

\begin{fact}[Restriction of Solutions for Positive Nested Conditions]\label{fact:satisfAC}
Given a positive nested condition $ac_{P_A}$ and a match $p_A: P_A \rightarrow G_A$ over $TG_A$ with restrictions 
$ac_{P_B} = Restr_t(ac_{P_A})$, $p_B = Restr_t(p_A)$ along $t: TG_B \rightarrow TG_A$. Then for a solution $Q_A$ of $p_A \vDash ac_{P_A}$, 
there is a solution $Q_B = Restr_t(Q_A)$ %a solution of 
for
$p_B \vDash ac_{P_B}$.
\end{fact}
%
%\begin{proof}
%See Appendix \ref{fact:satisfACAppendix}.
%\end{proof}

\section{Amalgamation}
\label{sec:amalgamation}
The amalgamation of typed objects allows to combine objects of different types provided that they agree on a common subtype. 
This concept is already known in the context of different types of Petri net processes, such as open net processes \cite{BCEH01}
and algebraic high-level processes \cite{EG11}, which can be seen as special kinds of typed objects. 
In this section, we introduce a general definition for the amalgamation of typed objects. Moreover, we extend the concept to 
the amalgamation of positive nested conditions and their solutions. 

As  required for amalgamation, we discuss under which conditions morphisms can be composed via a span of restriction morphisms. 
Two morphisms $g_B$ and $g_C$ ``agree'' in a morphism $g_D$, if $g_D$ can be constructed as a common restriction and can be used 
as a composition interface for $g_B$ and $g_C$ as in \autoref{def:agreement}.

\begin{definition}[Agreement and Amalgamation of Typed Objects]\label{def:agreement}
Given a span $TG_B \stackrel{tg_{DB}}{\longleftarrow} TG_D \stackrel{tg_{DC}}{\longrightarrow} TG_C$, 
with $tg_{DB}, tg_{DC} \in \M$ and typed objects $G_B \stackrel{g_{B}}{\rarr} TG_B$, $G_C \stackrel{g_{C}}{\rarr} TG_C$ and $G_D \stackrel{g_{D}}{\rarr} TG_D$. 
We say $g_{B},g_{C}$ \emph{agree} in $g_{D}$, if $g_{D}$ is a restriction of $g_{B}$ and $g_{C}$, i.e.,~$Restr_{tg_{DB}}(g_{B}) = g_{D} = Restr_{tg_{DC}}(g_{C})$.

Given pushout (1) below with all morphisms in \M and typed objects $g_B, g_C$ agreeing in $g_D$. 
A morphism $g_A : G_A \rarr TG_A$ is called \emph{amalgamation} of $g_B$ and $g_C$ over $g_D$, written $g_A = g_B +_{g_D} g_C$, 
if the outer square is a pushout and $g_B, g_C$ are restrictions of $g_A$.

	\xcmatrix{@R-3ex@C9ex}{
		& & G_D \ar[d]^{g_D} \ar[ddll] \ar[ddrr]
		& &\\
		& & TG_D \ar[dl]^{tg_{DB}} \ar[dr]_{tg_{DC}}
		& & \\
		G_B \ar[r]_{g_B} \ar[ddrr]
		& TG_B \ar[dr]^{tg_{BA}} \ar@{}[rr]|{(1)}
		& 
		& TG_C \ar[dl]_{tg_{CA}}
		& G_C \ar[l]^{g_C} \ar[ddll] \\
		& & TG_A 
		& & \\
		& & G_A \ar[u]^{g_A}
	}
%	\xcmatrix{@R-3ex@C9ex}{
%	& & G_D \ar[d]^{g_D} \ar[ddll] \ar[ddrr]
%	& &\\
%	& & TG_D \ar[dl]^{tg_{DB}} \ar[dr]_{tg_{DC}}
%	& & \\
%	G_B \ar[r]_{g_B}
%	& TG_B
%	&
%	& TG_C
%	& G_C \ar[l]^{g_C}
%}
\end{definition}

%\begin{definition}[Amalgamation of Typed Objects]\label{def:amalgamation}
%Given pushout (1) below with all morphisms in \M and typed objects $g_B, g_C$ agreeing in $g_D$. 
%A morphism $g_A : G_D \rarr TG_A$ is called \emph{amalgamation} of $g_B$ and $g_C$ over $g_D$, written $g_A = g_B +_{g_D} g_C$, 
%if the outer square is a pushout and $g_B, g_C$ are restrictions of $g_A$.
%
%	\xcmatrix{@R-3ex@C9ex}{
%		& & G_D \ar[d]^{g_D} \ar[ddll] \ar[ddrr]
%		& &\\
%		& & TG_D \ar[dl]^{tg_{DB}} \ar[dr]_{tg_{DC}}
%		& & \\
%		G_B \ar[r]_{g_B} \ar[ddrr]
%		& TG_B \ar[dr]^{tg_{BA}} \ar@{}[rr]|{(1)}
%		& 
%		& TG_C \ar[dl]_{tg_{CA}}
%		& G_C \ar[l]^{g_C} \ar[ddll] \\
%		& & TG_A 
%		& & \\
%		& & G_A \ar[u]^{g_A}
%	}
%\end{definition}

\autoref{fact:amalgamation} is essentially based on the horizontal VK property.

\begin{fact}[Amalgamation of Typed Objects]\label{fact:amalgamation}
Given pushout (1) with all morphisms in \M as in \autoref{def:agreement}.
\begin{description} 
	\item \textbf{Composition.}
		Given $g_B,g_C$ agreeing in $g_D$, then there exists a unique amalgamation $g_A = g_B +_{g_D} g_C$. 
	\item \textbf{Decomposition.}
		Vice versa, given $g_A: G_A \rarr TG_A$, there are unique restrictions $g_B, g_C,$ and $g_D$ of $g_A$ such that $g_A = g_B +_{g_D} g_C$.
\end{description}
Here and in the following, uniqueness means uniqueness up to isomorphism.
\end{fact}

\begin{proof}
Given $g_B, g_C$ agreeing in $g_D$, we have that the upper two trapezoids are pullbacks. Now we construct $G_A$ as pushout over $G_B$ and $G_C$ via $G_D$, such that the outer diamond is a pushout. This leads to a unique induced morphism $g_A: G_A \rarr TG_A$, such that the diagram commutes and via the horizontal VK property we get that the lower two trapezoids are pullbacks and therefore $g_A = g_B +_{g_D} g_C$.

Vice versa, we can construct $G_B, G_C, G_D$ as restrictions such that the trapezoids become pullbacks, where $g_A: G_A \rarr TG_A$ and $TG_A, TG_B, TG_C, TG_D$ are given such that (1) is a pushout with \M-morphisms only. Then the horizontal VK property implies that the outer diamond is a pushout and $g_A$ is unique because of the universal property and $g_A = g_B +_{g_D} g_C$.

The uniqueness (up to isomorphism) of the amalgamated composition and decomposition constructions follows from uniqueness of pushouts
and pullpacks up to isomorphism.
\end{proof}

\begin{example}[Amalgamation of Typed Objects]\label{ex:amalgamation-objects}
	Figure~\ref{fig:amalgamation-objects} shows a pushout of type graphs $TG_A$, $TG_B$, $TG_C$ and $TG_D$. 
	\begin{description}
		\item \textbf{Composition.}
			Consider the typed graphs $G_B$, $G_C$ and $G_D$ typed over $TG_B$, $TG_C$ and $TG_D$, respectively. Graph $G_D$, containing 
			the same nodes as $G_B$ and $G_C$ and no edges, is the common restriction of $G_B$ and $G_C$. So, the type morphisms $g_B$ and $g_C$
			agree in $g_D$, which by \autoref{fact:amalgamation} means that there is an amalgamation $g_A = g_B +_{g_D} g_C$. It can be obtained
			by computing the pushout of $G_B$ and $G_C$ over $G_D$, leading to the graph $G_A$ that contains the \texttt{b}-edges of $G_B$ as well
			as the \texttt{c}-edges of $G_C$. The type morphism $g_A$ is induced by the universal property of pushouts, mapping all edges
			in the same way as $g_B$ and $g_C$.
		\item \textbf{Decomposition.}
			Vice versa, consider the graph $G_A$ typed over $TG_A$. We can restrict $G_A$ to the type graphs $TG_B$ and $TG_C$, leading
			to typed graphs $G_B$ and $G_C$, containing only the \texttt{b}- respectively \texttt{c}-edges of $G_A$. Restricting the graphs
			$G_B$ and $G_C$ to type graph $TG_D$, we get in both cases the graph $G_D$ that contains no edges, and we have that $g_A = g_B +_{g_D} g_C$.
	\end{description}
	
	\begin{figure}[htb]%
	\centering
	\includegraphics[width=0.5\textwidth]{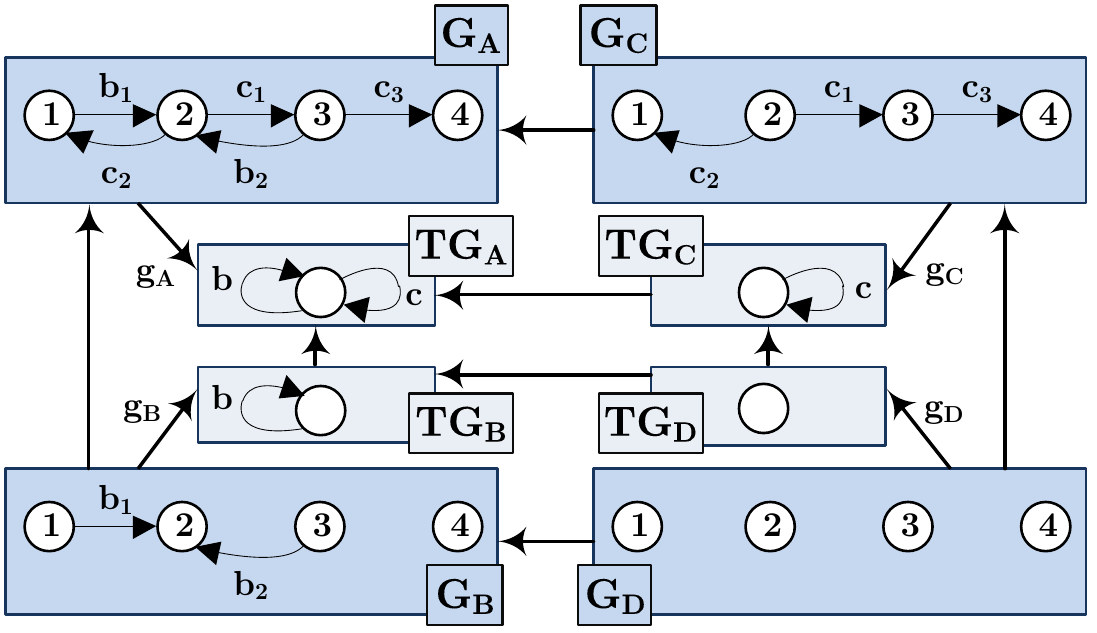}
	\caption{Amalgamation of typed graphs}%
	\label{fig:amalgamation-objects}
	\end{figure}
\end{example}

We already defined the restriction of positive nested conditions (\autoref{def:restrAC}) and their solutions \linebreak (\autoref{def:restrSolution}). Now we want to consider the case that we have two conditions, which have a common restriction and can be amalgamated.

\begin{definition}[Agreement and Amalgamation of Positive Nested Conditions]\label{def:agreement-amalgamation-conditions}
	Given a pushout (1) below with all morphisms in \M. 
	Two positive nested conditions $ac_{P_B}$ typed over $\TG_B$ and $ac_{P_C}$ typed over $\TG_C$ \emph{agree} in $ac_{P_D}$
	typed over $\TG_D$ if $ac_{P_D}$ is a restriction of $ac_{P_B}$ and $ac_{P_C}$.
	
	Given $ac_{P_B}$ and $ac_{P_C}$ agreeing in $ac_{P_D}$ then a positive nested condition 
	$ac_{P_A}$ typed over $TG_A$ is called \emph{amalgamation} of $ac_{P_B}$ and $ac_{P_C}$ over
	$ac_{P_D}$, written $ac_{P_A} = ac_{P_B} +_{ac_{P_D}} ac_{P_C}$, if $ac_{P_B}$ and $ac_{P_C}$ are restrictions of $ac_{P_A}$
	and $t_{PA} = t_{PB} +_{t_{PD}} t_{PC}$.
	In particular, we have $true_A = true_B +_{true_D} true_C$, short $true = true +_{true} true$.

	\xcmatrix{@R-3ex@C9ex}{
		& & P_D \ACright{ac_{P_D}} \ar[d]^{t_{PD}} \ar[ddll] \ar[ddrr]
		& &\\
		& & TG_D \ar[dl]^{tg_{DB}} \ar[dr]_{tg_{DC}}
		& & \\
		P_B \ACleft{ac_{P_B}} \ar[r]_{t_{PB}} \ar[ddrr]
		& TG_B \ar[dr]^{tg_{BA}} \ar@{}[rr]|{(1)}
		& 
		& TG_C \ar[dl]_{tg_{CA}}
		& P_C \ACright{ac_{P_C}} \ar[l]^{t_{PC}} \ar[ddll] \\
		& & TG_A 
		& & \\
		& & P_A \ACright[4.5ex,-0.5ex]{ac_{P_A}} \ar[u]^{t_{PA}}
	}
\end{definition}

In the following \autoref{fact:amalgamation-conditions}, we give a construction for the amalgamation of positive nested conditions and in \autoref{thm:ACviaRestr} for the corresponding solutions.

\begin{fact}[Amalgamation of Positive Nested Conditions]\label{fact:amalgamation-conditions}
	Given a pushout (1) as in \autoref{def:agreement-amalgamation-conditions} with all morphisms in \M. 
	\begin{description}
		\item \textbf{Composition.} 
			If there are positive nested conditions $ac_{P_B}$ and $ac_{P_C}$ typed over $TG_B$ and $TG_C$, respectively,
			agreeing in $ac_{P_D}$ typed over $TG_D$, then there exists a unique positive nested condition $ac_{P_A}$ typed over $TG_A$ such that 
			$ac_{P_A} = ac_{P_B} +_{ac_{P_D}} ac_{P_C}$.
		\item \textbf{Decomposition.}
			Vice versa, given a positive nested condition $ac_{P_A}$ typed over $TG_A$, there are unique restrictions $ac_{P_B}$, $ac_{P_C}$ and $ac_{P_D}$ 
			of $ac_{P_A}$ such that $ac_{P_A} = ac_{P_B} +_{ac_{P_D}} ac_{P_C}$.
	\end{description}
	The amalgamated composition and decomposition constructions are unique up to isomorphism.

\end{fact}

\begin{remark}\label{rem:amalgamation-conditions}
	Given an amalgamation $ac_{P_A} = ac_{P_B} +_{ac_{P_D}} ac_{P_C}$ of positive nested conditions, we can conclude from the proof of \autoref{fact:amalgamation-conditions} (see \autoref{sec:appendixB}) that we also have corresponding amalgamations in each
	level of nesting.
\end{remark}

%%% belongs to example below	
\begin{figure}[htb]%
\centering
\includegraphics[width=0.9\textwidth]{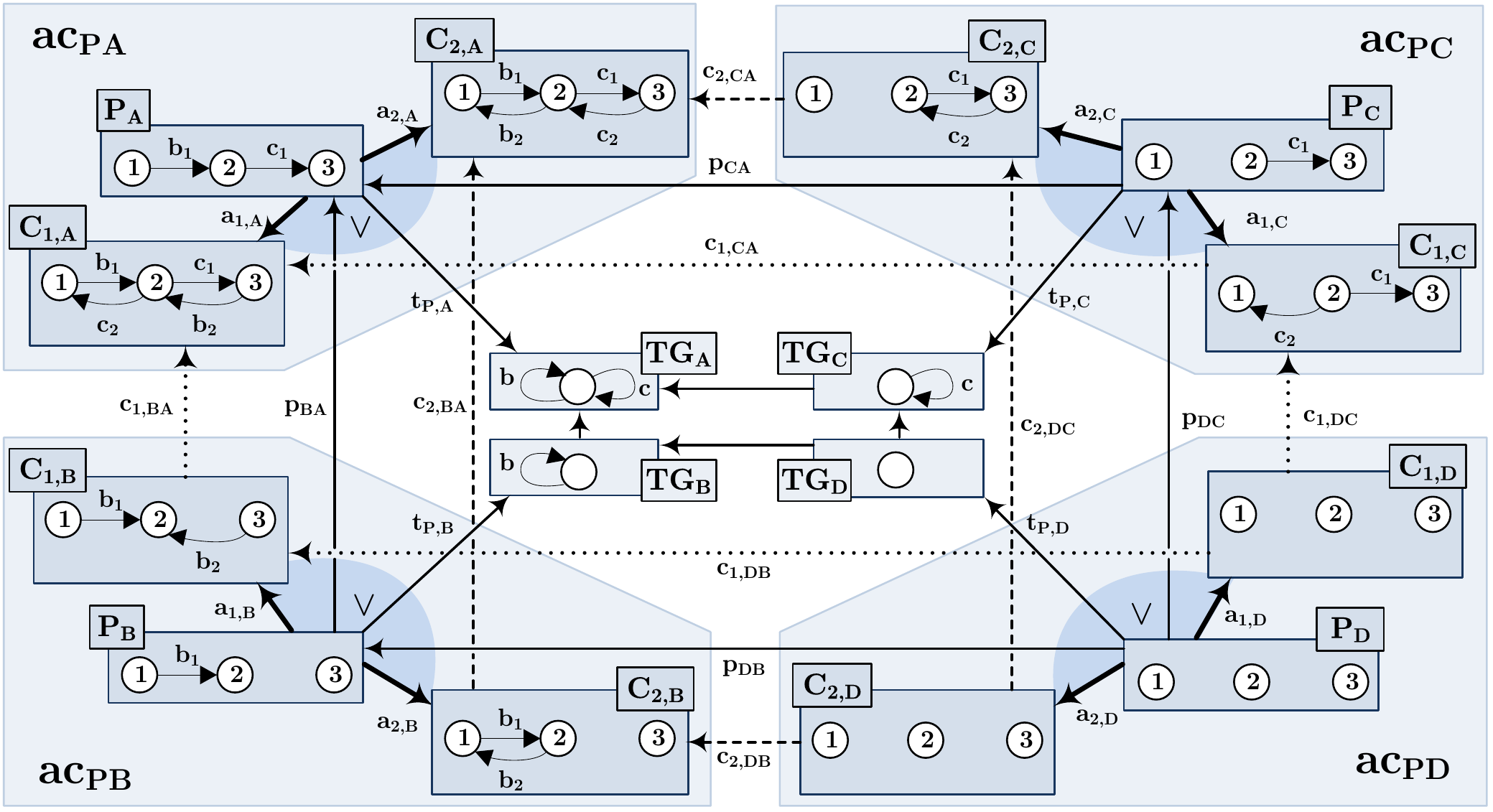}
\caption{Amalgamation of positive nested conditions}%
\label{fig:amalgamation-conditions}
\end{figure}

\begin{example}[Amalgamation of Positive Nested Conditions]\label{ex:amalgamation-conditions}
	Figure~\ref{fig:amalgamation-conditions} shows a pushout of typed graphs $TG_A$, $TG_B$, $TG_C$ and $TG_D$, and four positive nested conditions
	$ac_{P_A}$, $ac_{P_B}$, $ac_{P_C}$ and $ac_{P_D}$ typed over $TG_A$, $TG_B$, $TG_C$ and $TG_D$, respectively. For simplicity, the figure
	contains only the type morphisms of the $P$s, but there are also corresponding type morphisms for the $C$s, mapping all \texttt{b}-edges to
	\texttt{b} and all \texttt{c}-edges to \texttt{c}. 
	There is $ac_{P_A} = \bigvee_{i \in \{1,2\}} ac_{C_{i,A}}$ with $ac_{C_{i,A}} = \exists(a_{i,A}, true)$ for $i = 1,2$, and $ac_{P_B}$,
	$ac_{P_C}$ and $ac_{P_D}$ have a similar structure.
	\begin{description}
		\item \textbf{Composition.}
			We have that $t_{P_D}$ is a common restriction of $t_{P_B}$ and $t_{P_C}$, and also that $a_{i,D}$ is a common
			restriction of $a_{i,B}$ and $a_{i,C}$ for $i=1,2$. Thus, $ac_{P_D}$ is a common restriction of $ac_{P_B}$ and $ac_{P_C}$
			which means that $ac_{P_B}$ and $ac_{P_C}$ agree in $ac_{P_D}$. So by \autoref{fact:amalgamation-conditions} there
			exists an amalgamation $ac_{P_A} = ac_{P_B} +_{ac_{P_D}} ac_{P_C}$, and according to \autoref{rem:amalgamation-conditions}
			it can be obtained as amalgamation of its components. This means that we have an amalgamation 
			$t_{P_A} = t_{P_B} +_{t_{P_D}} t_{P_C}$ with pushout of the $P$s as shown in \autoref{fig:amalgamation-conditions},
			as well as amalgamations of the corresponding type morphisms of the $C$s, leading to the
			pushouts depicted in \autoref{fig:amalgamation-conditions} by dotted arrows for the $C_1$s and by dashed arrows for the $C_2$s.
			The morphisms $a_{1,A}$ and $a_{2,A}$ are obtained by the universal property of pushouts.
		\item \textbf{Decomposition.}
			The other way around, considering the condition $ac_{P_A}$, we can construct the restrictions $ac_{P_B}$ and $ac_{P_C}$ by deleting
			the \texttt{c}- respectively \texttt{b}-edges. Then, restricting $ac_{P_B}$ and $ac_{P_C}$ to $TG_D$ by deleting all remaining edges,
			we obtain the same condition $ac_{P_D}$ such that $ac_{P_A} = ac_{P_B} +_{ac_{P_D}} ac_{P_C}$.
	\end{description}
\end{example}

In order to answer the question, under which conditions such amalgamated positive nested conditions are satisfied, we need to define an amalgamation of their solutions. Afterwards, we show in the proof of \autoref{thm:ACviaRestr} that a composition of two solutions via an interface leads to a unique amalgamated solution and that a given solution for an amalgamated positive nested condition is the amalgamation of its unique restrictions.

\begin{definition}[Agreement and Amalgamation of Solutions for Positive Nested Conditions]\label{def:agreement-amalgamation-solution}
	Given pushout (1) below with all morphisms in \M, an amalgamation of typed objects $g_A = g_B +_{g_D} g_C$,
	and an amalgamation of positive nested conditions $ac_{P_A} = ac_{P_B} +_{ac_{P_D}} ac_{P_C}$ 
	with corresponding matches $p_A = p_B +_{p_D} p_C$.
	\begin{enumerate}
		\item Two solutions $Q_B$ for $p_B \vDash ac_{P_B}$ and $Q_C$ for $p_C \vDash ac_{P_C}$ \emph{agree} in a solution $Q_D$ for $p_D \vDash ac_{P_D}$,
			if $Q_D$ is a restriction of $Q_B$ and $Q_C$.
		
		\item Given solutions $Q_B$ for $p_B \vDash ac_{P_B}$ and $Q_C$ for $p_C \vDash ac_{P_C}$ agreeing in a solution $Q_D$ for $p_D \vDash ac_{P_D}$, 
			then a solution $Q_A$ for $p_A \vDash ac_{P_A}$ is called \emph{amalgamation} of
			$Q_B$ and $Q_C$ over $Q_D$, written $Q_A = Q_B +_{Q_D} Q_C$, if $Q_B$ and $Q_C$ are restrictions of $Q_A$.
	\end{enumerate}

	\xcmatrix{@R-3ex}{
		& P_A \ACleft{ac_{P_A}} \ar[dr]_{p_A}
		&&&& 
		& P_C \ACright{ac_{P_C}} \ar[lllll]^{p_{CA}} \ar[dl]^{p_C}
		& \\
		&& G_A \ar[dr]_{g_A}
		&&
		& G_C \ar[lll]^{g_{CA}} \ar[dl]^{g_C}
		&& \\
		&&& TG_A \ar@{}[dr]|{(1)}
		& TG_C \ar[l]_{tg_{CA}}
		&&& \\
		&&& TG_B \ar[u]^{tg_{BA}}
		& TG_D \ar[l]_{tg_{DB}} \ar[u]_{tg_{DC}}
		&&& \\
		&& G_B \ar[uuu]_{g_{BA}} \ar[ur]^{g_B}
		&&
		& G_D \ar[lll]^{g_{DB}} \ar[uuu]_{g_{DC}} \ar[ul]_{g_D}
		&& \\
		& P_B \ACleft{ac_{P_B}} \ar[uuuuu]_{p_{BA}} \ar[ur]^{p_B}
		&&&& 
		& P_D \ACright{ac_{P_D}} \ar[lllll]^{p_{DB}} \ar[uuuuu]_{p_{DC}} \ar[ul]_{p_D}
		& 
	}
\end{definition}

\begin{remark}\label{rem:agreement-amalgamation-solution}
	Note that by assumption $g_A = g_B +_{g_D} g_C$ in the definition above we already have a pushout over the $G$s, and by 
	$ac_{P_A} = ac_{P_B} +_{ac_{P_D}} ac_{P_C}$ we also have a pushout over the $P$s. 
\end{remark}

\begin{theorem}[Amalgamation of Solutions for Positive Nested Conditions]\label{thm:ACviaRestr}
Given pushout (1) as in \autoref{def:agreement-amalgamation-solution} with all morphisms in \M, 
an amalgamation of typed objects $g_A = g_B +_{g_D} g_C$,
and an amalgamation of positive nested conditions $ac_{P_A} = ac_{P_B} +_{ac_{P_D}} ac_{P_C}$ 
with corresponding matches $p_A = p_B +_{p_D} p_C$.
\begin{description}
	\item \textbf{Composition.}
		Given solutions $Q_B$ for $p_B \vDash ac_{P_B}$ and $Q_C$ for $p_C \vDash ac_{P_C}$ agreeing in a solution $Q_D$ for $p_D \vDash ac_{P_D}$,
		then there is a solution $Q_A$ for $p_A \vDash ac_{P_A}$ constructed as 
		amalgamation $Q_A = Q_B +_{Q_D} Q_C$.
	\item \textbf{Decomposition.}
		Given a solution $Q_A$ for $p_A \vDash ac_{P_A}$, then there are 
		solutions $Q_B$, $Q_C$ and $Q_D$ for \linebreak
		$p_B \vDash ac_{P_B}$, $p_C \vDash ac_{P_C}$ 
		and $p_D \vDash ac_{P_D}$, respectively, which are constructed as
		restrictions $Q_B$, $Q_C$ and $Q_D$ of $Q_A$ such that
		$Q_A = Q_B +_{Q_D} Q_C$.
\end{description}
The amalgamated composition and decomposition constructions are unique up to isomorphism.
%
%Given pushout (1) as in Def.~\ref{def:agreement-amalgamation-solution} with all morphisms in \M, 
%an amalgamation of typed objects $g_A = g_B +_{g_D} g_C$,
%and an amalgamation of positive nested conditions $ac_{P_A} = ac_{P_B} +_{ac_{P_D}} ac_{P_C}$ 
%with corresponding matches $p_A = p_B +_{p_D} p_C$.
%\begin{description}
%	\item \textbf{Composition.}
%		Given solutions $Q_B$ for $p_B \vDash ac_{P_B}$ and $Q_C$ for $p_C \vDash ac_{P_C}$ agreeing in a solution $Q_D$ for $p_D \vDash ac_{P_D}$,
%		then there exists a unique solution $Q_A$ for $p_A \vDash ac_{P_A}$ such that $Q_A = Q_B +_{Q_D} Q_C$.
%	\item \textbf{Decomposition.}
%		Given a solution $Q_A$ for $p_A \vDash ac_{P_A}$, then there exist unique restrictions $Q_B$, $Q_C$ and $Q_D$ such that
%		$Q_A = Q_B +_{Q_D} Q_C$.
%\end{description}
%The amalgamated composition and decomposition constructions are unique up to isomorphism.
\end{theorem}

\begin{remark}\label{rem:amalgamation-solutions}
	From the proof of \autoref{thm:ACviaRestr} (see \autoref{sec:appendixB}) we can conclude that for a given amalgamation of solutions $Q_A = Q_B +_{Q_D} Q_C$, we also
	have corresponding amalgamations of its components.
\end{remark}

\section{Compatibility of Initial Satisfaction with Restriction and Amalgamation}
\label{sec:compatibility}
In this section we present our main result showing compatibility of initial satisfaction with amalgamation (\autoref{thm:compatibility-initial-satisfaction-amalgamation}) and restriction (\autoref{cor:compatibility-initial-satisfaction-restriction}) which are based on the amalgamation of solutions for positive nested conditions (\autoref{thm:ACviaRestr}). This main result allows to conclude the satisfaction of a constraint for a composed object from the satisfaction of the corresponding restricted constraints for the component objects. It is valid for initial satisfaction, but not for general satisfaction.

\begin{theorem}[Compatibility of Initial Satisfaction with Amalgamation]
\label{thm:compatibility-initial-satisfaction-amalgamation}
	Given pushout (1) below with all morphisms in \M, 
	an amalgamation of typed objects $g_A = g_B +_{g_D} g_C$, and an amalgamation of positive constraints $ac_A = ac_B +_{ac_D} ac_C$.
	Then we have:
\begin{description}
	\item \textbf{Decomposition.} Given a solution $Q_A$ for $G_A \iDash ac_A$, then there are solutions $Q_B$ for $G_B \iDash ac_B$, $Q_C$ for $G_C \iDash ac_C$
			and $Q_D$ for $G_D \iDash ac_D$ such that $Q_A = Q_B +_{Q_D} Q_C$.
	\item \textbf{Composition.} Vice versa, given solutions $Q_B$ for $G_B \iDash ac_B$ and $Q_C$ for $G_C \iDash ac_C$
			agreeing in a solution $Q_D$ for $G_D \iDash ac_D$, then there exists a solution $Q_A$ for $G_A \iDash ac_A$ such that 
			$Q_A = Q_B +_{Q_D} Q_C$.
\end{description}
	
	\xcmatrix{@R-3ex@C+2ex}{
		& I \ACleft{ac_A} \ar[dr]^(.6){i_{G_A}} 
		&&&& 
		& I  \ACright{ac_C} \ar[lllll]^{id_I} \ar[dl]_(.6){i_{G_C}} 
		& \\
		&& G_A \ar[dr]_{g_A} 
		&&
		& G_C \ar[lll]^{g_{CA}} \ar[dl]^{g_C}
		&& \\
		&&& TG_A \ar@{}[dr]|{(1)}
		& TG_C \ar[l]_{tg_{CA}}
		&&& \\
		&&& TG_B \ar[u]^{tg_{BA}}
		& TG_D \ar[l]^{tg_{DB}} \ar[u]_{tg_{DC}}
		&&& \\
		&& G_B \ar[uuu]_{g_{BA}} \ar[ur]^{g_B}
		&&
		& G_D \ar[lll]_{g_{DB}} \ar[uuu]^{g_{DC}} \ar[ul]_{g_D}
		&& \\
		& I \ACleft{ac_B} \ar[uuuuu]_{id_I} \ar[ur]_(.6){i_{G_B}}
		&&&& 
		& I \ACright{ac_D} \ar[lllll]_{id_I} \ar[uuuuu]^{id_I} \ar[ul]^(.6){i_{G_D}}
		& 
	}

\begin{proof}\leavevmode
\begin{description}
	\item \textbf{Decomposition.} By \autoref{def:initSatisfaction} a solution $Q_A$ for $G_A \iDash ac_A$ is also a solution for $i_{G_A} \vDash ac_A$,
				where $i_{G_A}$ is the unique morphism $i_{G_A}: I \to G_A$. Moreover, due to amalgamation $g_A = g_B +_{g_D} g_C$ the inner trapezoids
				in the diagram above are pullbacks. So by closure of \M under pullbacks we have that $g_{BA}, g_{CA}, g_{DB}, g_{DC} \in \M$
				which means that they are monomorphisms. Therefore, the outer trapezoids become pullbacks by standard category theory, which means
				that $i_{G_B}: I \to G_B$ is a restriction of $i_{G_A}$, $i_{G_C}: I \to G_C$ is a restriction of $i_{G_A}$, and
				$i_{G_D}: I \to G_D$ is a restriction of $i_{G_B}$ as well as of $i_{G_C}$. 
				
				Furthermore, the outer square in the diagram is a pushout, implying that we have an amalgamation $i_{G_A} = i_{G_B} +_{i_{G_D}} i_{G_C}$.
				Thus, using \autoref{thm:ACviaRestr} we obtain solutions $Q_B$ for $i_{G_B} \vDash ac_B$, $Q_C$ for $i_{G_C} \vDash ac_C$ 
				and $Q_D$ for $i_{G_D} \vDash ac_D$ such that $Q_A = Q_B +_{Q_D} Q_C$, and by \autoref{def:initSatisfaction}
				$Q_B$, $Q_C$ and $Q_D$ are solutions for $G_B \iDash ac_B$, $G_C \iDash ac_C$ and $G_D \iDash ac_D$, respectively.
			
			\item \textbf{Composition.} Now, given solutions $Q_B$, $Q_C$ and $Q_D$ for $G_B \iDash ac_B$, $G_C \iDash ac_C$ and $G_D \iDash ac_D$, respectively. Then by
				\autoref{def:initSatisfaction} we have that $Q_B$, $Q_C$ and $Q_D$ are solutions for $i_{G_B} \vDash ac_B$,
				$i_{G_C} \vDash ac_C$ and $i_{G_D} \vDash ac_D$, respectively. As shown in item 1, there is $i_{G_A} = i_{G_B} +_{i_{G_D}} i_{G_C}$
				and therefore, since $Q_B$ and $Q_C$ agree in $Q_D$, by \autoref{thm:ACviaRestr} we obtain a solution $Q_A$ for $i_{G_A} \vDash ac_A$
				such that $Q_A = Q_B +_{Q_D} Q_C$. Finally, \autoref{def:initSatisfaction} implies that $Q_A$ is a solution for $G_A \iDash ac_A$.
\end{description}
\end{proof}
\end{theorem}

\begin{corollary}[Compatibility of Initial Satisfaction with Restriction]\label{cor:compatibility-initial-satisfaction-restriction}
Given type restriction \linebreak $t: TG_B \rarr TG_A \in \M$, object $G_A$ typed over $TG_A$ with restriction $G_B$, and a positive 
constraint $ac_A$ over initial object $I$ typed over $TG_A$ with restriction $ac_B$. Then $G_A \iDash ac_A$ 
implies $G_B \iDash ac_B$. Moreover, if $Q_A$ is a solution for $G_A \iDash ac_A$ then $Q_B = Restr_t(Q_A)$ is a solution
for $G_B \iDash ac_B$.

\begin{proof}
	Consider the diagram in \autoref{thm:compatibility-initial-satisfaction-amalgamation} with $G_C = G_A$, $G_D = G_B$, $ac_C = ac_A$
	and $ac_D = ac_B$. Then by standard category theory we have that all rectangles in the diagram are pushouts and the trapezoids are pullbacks.
	Thus, we have $g_A = g_B +_{g_B} g_A$ and, analogously, $ac_A = ac_B +_{ac_B} ac_A$ 
	with corresponding matches $i_{G_A} = i_{G_B} +_{i_{G_B}} i_{G_A}$.
	So, given a solution $Q_A$ for $G_A \iDash ac_A$, by item 1 of \autoref{thm:compatibility-initial-satisfaction-amalgamation} there
	is a solution $Q_B$ for $G_B \iDash ac_B$ with $Q_A = Q_B +_{Q_B} Q_A $ such that by \autoref{def:agreement-amalgamation-conditions} 
	$Q_B$ is a restriction of $Q_A$.
\end{proof}
\end{corollary}

\begin{example}[Compatibility of Initial Satisfaction with Amalgamation]\label{ex:initial-satisfaction-amalgamation}
	Figure~\ref{f:amalgamation} shows the amalgamation of typed graphs $g_A = g_B +_{g_D} g_C$ from \autoref{ex:amalgamation-objects}
	and an amalgamation of positive nested conditions $ac_A = ac_B +_{ac_D} ac_C$. Note that we have $ac_A = \exists(i_{P_A}, ac_{P_A})$
	and $ac_B$, $ac_C$ and $ac_D$ with similar structure, where the amalgamation $ac_{P_A} = ac_{P_B} +_{ac_{P_D}} ac_{P_C}$ is presented in
	\autoref{ex:amalgamation-conditions}. 
	\begin{description}
		\item \textbf{Composition.} For $G_B \iDash ac_B$ we have the solution $Q_B = (q_B,(Q_{1,B}, Q_{2,B}))$ 
		with $Q_{1,B} = (q_{1,B},\emptyset)$ and $Q_{2,B} = \emptyset$, where $q_B$ and $q_{1,B}$ are inclusions.
		Moreover, we have similar solutions $Q_C$ for \linebreak $G_C \iDash ac_C$ and $Q_D$ for $G_D \iDash ac_D$. 
		According to \autoref{rem:amalgamation-solutions}, the amalgamation $Q_A = Q_B +_{Q_D} Q_C$ can be constructed by amalgamation
		of the components. 
		
		First, we explain in detail the amalgamation $q_{1,A} = q_{1,B} +_{q_{1,D}} q_{1,C}$. 
		Note that the graphs $G_A$, $G_B$, $G_C$ and $G_D$ can be considered as type graphs such that e.\,g.\ $C_{1,D}$ is typed over $G_D$
		by $q_{1,D}$. So, since $q_{1,D}$ is a common restriction of $q_{1,B}$ and $q_{1,C}$, we have that $q_{1,B}$ and $q_{1,C}$
		agree in $q_{1,D}$. This means that there is an amalgamation of typed objects $q_{1,A} = q_{1,B} +_{q_{1,D}} q_{1,C}$, where the inclusion
		$q_{1,A}$ maps all nodes and edges in the same way as $q_{1,B}$ and $q_{1,C}$.
		
		Moreover, for the empty solutions we have an empty solution as amalgamation, and thus, we have amalgamations of solutions
		$Q_{1,A} = Q_{1,B} +_{Q_{1,D}} Q_{1,C}= (q_{1,A}, \emptyset)$ and $Q_{2,A} = Q_{2,B} +_{Q_{2,D}} Q_{2,C}= \emptyset$.
		The amalgamation $q_A = q_B +_{q_D} q_C$ can be obtained analogously as described for $q_{1,A}$, and hence, we
		have $Q_A = Q_B +_{Q_D} Q_C = (q_A, (Q_{1,A}, Q_{2,A}))$, which is a solution for $G_A \iDash ac_A$.
		
		\item \textbf{Decomposition.} For $G_A \iDash ac_A$ we have a solution $Q_A = (q_A,(Q_{1,A}, Q_{2,A}))$ with $Q_{1,A} = (q_{1,A},\emptyset)$ 
		and $Q_{2,A} = \emptyset$ where $q_A$ and $q_{1,A}$ are inclusions. The restrictions $Q_B$, $Q_C$ and $Q_D$ of $Q_A$ are given by restrictions
		of the components. By computing the restrictions $q_{1,B}$, $q_{1,C}$ and $q_{1,D}$ of $q_{1,A}$, and similar the restrictions of $q_A$
		and $\emptyset$, we get as result again the solutions $Q_B$ for $G_B \iDash ac_B$, $Q_C$ for $G_C \iDash ac_C$,
		and $Q_D$ for $G_D \iDash ac_D$ as described in the composition case above.

	\end{description}
	\begin{figure}[htb]%
	\centering
	\includegraphics[width=\textwidth]{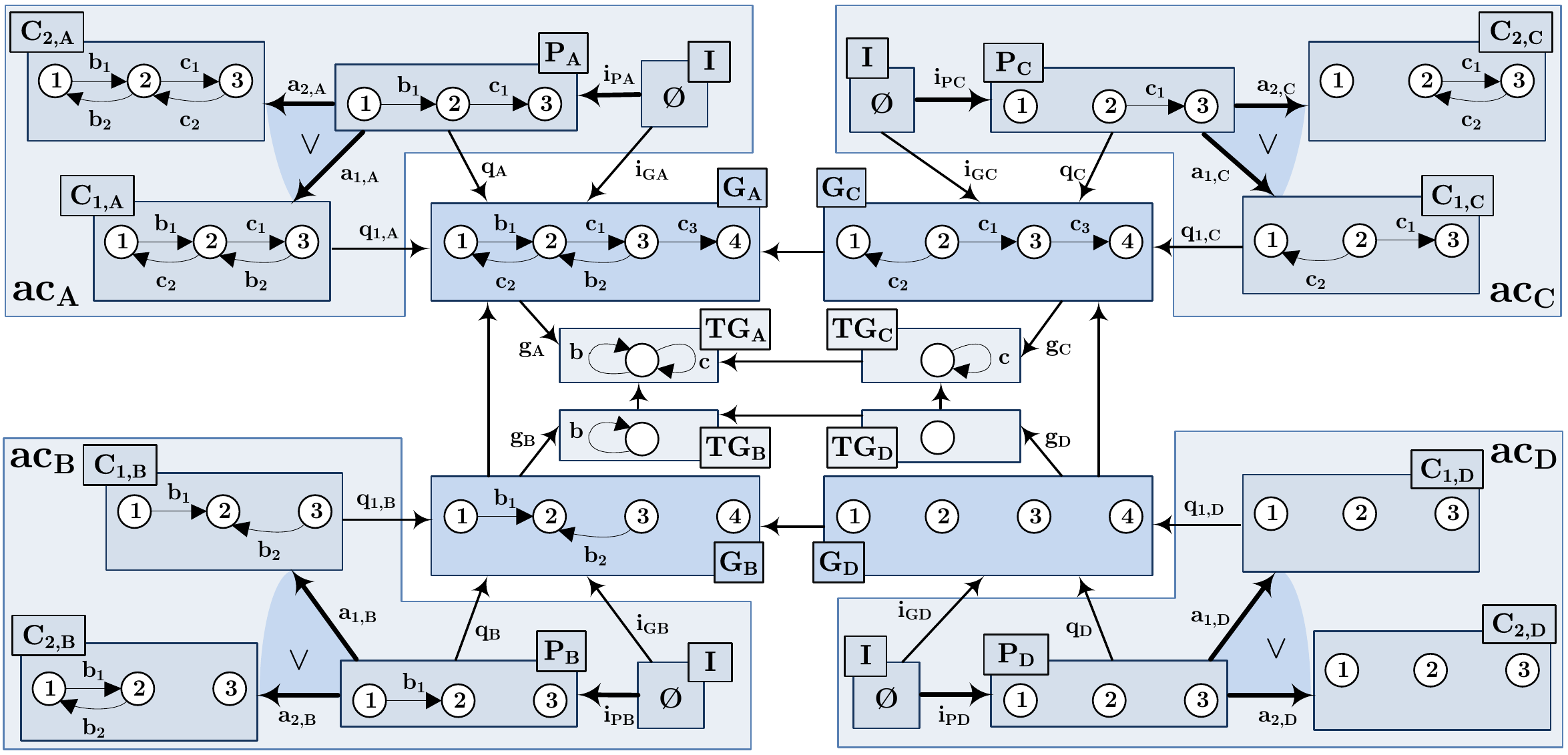}
	\caption{Amalgamation of solutions for initial satisfaction}%
	\label{f:amalgamation}
	\end{figure}
\end{example}

From \autoref{cor:compatibility-initial-satisfaction-restriction}, we know that initial satisfaction is compatible with restriction of
typed objects and constraints. In contrast, general satisfaction and restriction are not compatible in general. 
As the following example illustrates, it is possible that a typed object
generally satisfies a constraint while the same does not hold for their restrictions.

\begin{example}[Restriction of General Satisfaction Fails in General]
Figure~\ref{f:counterex} shows a restriction $G_B$ of the typed graph $G_A$ and a restriction $ac_{P_B}$ of constraint $ac_{P_A}$.
There are two possible matches $p_{1,A}, p_{2,A}: P_A \to G_A \in \M$ where $p_{1,A}$ is an inclusion and 
$p_{2,A}$ maps $\mathtt{b_1}$ to $\mathtt{b_2}$ and $\mathtt{c_1}$ to $\mathtt{c_2}$. Since for each of the matches the graph $G_A$
contains the required edges in the inverse direction, both of the matches satisfy $ac_{P_A}$. For $p_{i,A}$ we have $q_{i,A}$ with $q_{i,A} \circ a_A = p_{i,A}$ for $i = 1,2$. Thus, we have that $G_A \vDash ac_{P_A}$. 

For the constraint $ac_{P_B}$ there is a match $p_B: P_B \to G_B \in \M$ mapping edge $\mathtt{b_1}$ identically and node \texttt{3} to node \texttt{4}.
We have that $p_B \not\vDash ac_{P_B}$ because there is no edge from node \texttt{4} to node \texttt{2} in $G_B$, which means that $G_B \not\vDash ac_{P_B}$.
This is due to the fact that there is no match $p_A: P_A \to G_A \in \M$ such that $p_B$ is the restriction of $p_A$.

%\
%In Fig.~\ref{f:counterex} we have $G_A \models ac_A$ but $G_B \nvDash ac_B$, because for the match $p_B$, where $b_1$ is matched on $b_3$ there is no commutative extension $q_B$. This is due to the fact, that for $p_B$ there is noch extension to $p_A$ such that $p_A \circ t_B = t_G \circ p_B$.

%\rot{KG: The figure should also contain type graphs $TG_A$ and $TG_B$. Further, the labeling of the edges is inconsistent with the type graph.
%$G_B$ and $ac_B$ should contain $b$-edges instead of $c$-edges. The $c_3$-edge can be removed.} 

\begin{figure}[htb]%
\centering
\includegraphics[scale=.7]{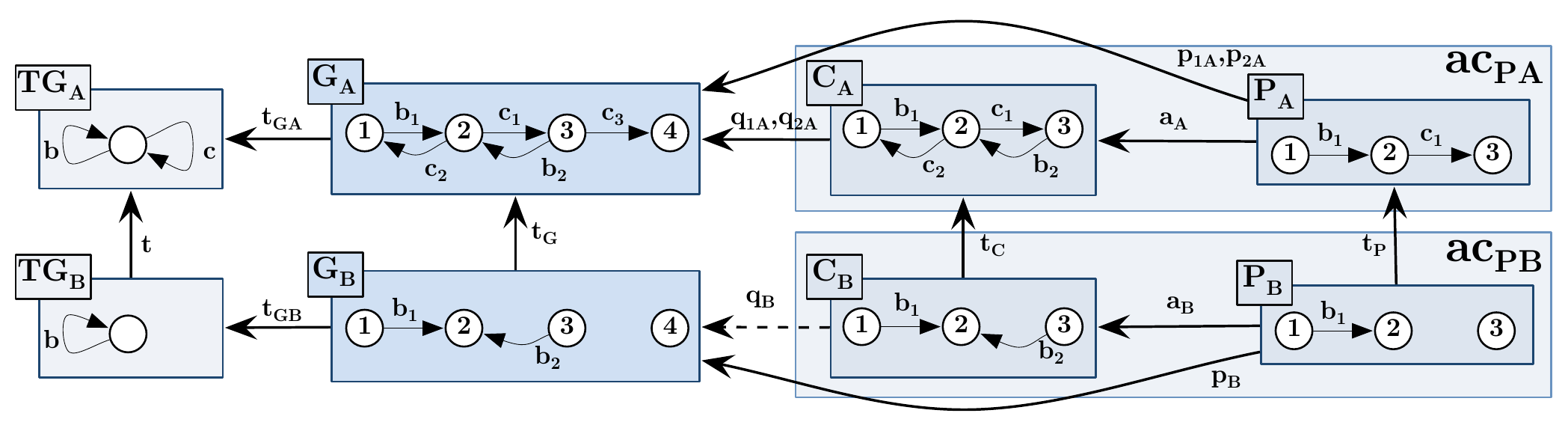}
\caption{Counterexample for restriction of general satisfaction}%
\label{f:counterex}
\end{figure}
\end{example}

\section{Related Work}
\label{sec:relatedWork}
The framework of $\M$-adhesive categories~\cite{EGH10} generalizes
various kinds of categories for high level replacement systems, e.g. 
adhesive~\cite{LS04}, quasi-adhesive~\cite{LS05}, partial VK square adhesive~\cite{Heindel10}, and weak-adhesive categories~\cite{EEPT06}.
Therefore, the results of this paper are applicable to all of them, where the category of typed attributed graphs is a prominent example.

The concepts of nested graph conditions~\cite{HP05} and first-order graph formulas~\cite{Courcelle97theexpression} are shown to be expressively equivalent in~\cite{HP09} using the translation between first-order logic and predicates on edge-labeled graphs without parallel edges~\cite{Rensink04}.

Multi-view modelling is an important concept in software engineering.
Several approaches have been studied and used, e.g. focussing on aspect oriented techniques~\cite{FranceRGG04}.
In this line, graph transformation (GT) approaches have been extended to support view concepts based on the integration of type graphs.
For this purpose, the concept of restriction along type morphisms has been studied and used intensively~\cite{EHTE97,EEEP10} including GT systems using the concept of inheritance and views~\cite{EEEP10,JT10}.
Instead of restriction of constraints considered in this paper, only more restrictive forward translations of view constraints have been studied in \cite{EEEP10} for the case of atomic constraints with general satisfaction leading to a result similar to \autoref{thm:compatibility-initial-satisfaction-amalgamation}.
%The restriction of nested conditions, however, has been studied only partially for atomic constraints~\cite{EEEP10}, and not yet for the more general nested conditions~\cite{HP09}, which are used in this paper.
The notions of initial and general satisfaction for nested conditions can be transformed one 
into the other~\cite{HP09}, but this transformation uses the Boolean operator negation that
is not present in positive constraints, for which, however, 
our main result on the compatibility of restriction and initial satisfaction holds.
Moreover, we have shown by counterexample that general satisfaction is not compatible with restriction in general, even if only positive constraints are considered.

\section{Conclusion}
\label{sec:conclusion}
Nested application conditions for rules and constraints for graphs and more general models have
been studied already in the framework of $\M$-adhesive transformation systems~\cite{EEPT06,EHPP06}.
The new contribution of this paper is to study compatibility of satisfaction with restriction and amalgamation.
This is important for large typed systems respectively objects, which can be decomposed by restriction and composed by amalgamation.
The main result in this paper shows that initial satisfaction of positive constraints is compatible with restriction and amalgamation (\autoref{thm:compatibility-initial-satisfaction-amalgamation} and \autoref{cor:compatibility-initial-satisfaction-restriction}).
The amalgamation construction is based on the horizontal van Kampen (VK) property, which is required in addition to the vertical VK property of $\M$-adhesive categories. To our best knowledge, this is the most interesting result for $\M$-adhesive transformation systems which is based on the horizontal VK property.
Note that the main result is not valid for general satisfaction of positive constraints nor for initial satisfaction of general constraints.
For future work, it is important to obtain weaker versions of the main result, which are valid for general satisfaction and constraints, respectively.

%
%Future work:
%\begin{itemize}
%  \item sufficient static conditions for compatibility with negation
%  \item application to concrete application domains
%	\item combination with model transformations and preservation of satisfaction of constraints: MTs for views
%\end{itemize}

\bibliographystyle{eptcs}
\bibliography{Bibliography}

%\FloatBarrier
%\section{Comments for Coauthors}
%\input{comments}

\newpage

\begin{appendix}
%\section{Properties of General and Initial Satisfaction}
%\label{sec:appendix}
%\input{2b-properties}
\section{Remaining Proofs}
\label{sec:appendixB}
%\subsection{Properties of General and Initial Satisfaction}
%\label{subsec:properties}
%\input{2b-properties}

In this appendix, we give the proofs for 
\autoref{fact:satisfAC}, 
\autoref{fact:amalgamation-conditions}
and \autoref{thm:ACviaRestr}.

%\begin{fact}[Restriction of Solutions for Positive Nested Conditions]\label{fact:satisfACAppendix}
\newtheorem*{factSatisfAC}{Fact~\ref*{fact:satisfAC}}
\begin{factSatisfAC}[Restriction of Solutions for Positive Nested Conditions]
Given a positive nested condition $ac_{P_A}$ and a match $p_A: P_A \rightarrow G_A$ over $TG_A$ with restrictions 
$ac_{P_B} = Restr_t(ac_{P_A})$, $p_B = Restr_t(p_A)$ along $t: TG_B \rightarrow TG_A$. Then for a solution $Q_A$ of $p_A \vDash ac_{P_A}$ 
there is a solution $Q_B = Restr_t(Q_A)$ %a solution of 
for
$p_B \vDash ac_{P_B}$.
%\end{fact}
\end{factSatisfAC}

\begin{proof}\leavevmode
\begin{itemize}
	\item 
		For $ac_{P_A} = true$ the implication is trivial, because $Q_A$ is empty which means that also $Q_B$ is empty and thus a solution for 
		$p_B \vDash ac_{P_B}$ is empty, because $ac_{P_B}$ is also $true$.
	\item
		For $ac_{P_A} = \exists (a,ac_{C_A})$ we have that $Q_A = (q_A, Q_{CA})$ such that $q_A : C_A \rightarrow G_A \in \M$ with $q_A \circ a = p_A$ 
		and $Q_{CA}$ is a solution for $q_A \vDash ac_{C_A}$. 
		Then by $q_B = Restr_t(q_A) : C_B \rightarrow G_B$, we have $q_B \in \M$ and 
		%by $p_A = Restr_t(p_B)$ 
		we also have		$t_G: G_B \to G_A \in \M$, because $t \in \M$  (see \autoref{fig:restriction-for-satisfaction}). So for $ac_{P_B} = \exists (b, ac_{C_B})$ we have
		\[ t_G \circ q_B \circ b = q_A \circ t_C \circ b = q_A \circ a \circ t_P = p_A \circ t_P = t_G \circ p_B, \]
		which by monomorphism $t_G$ implies $q_B \circ b = p_B$. 
		
		Moreover, the fact that $Q_{CA}$ is a solution for $q_A \vDash ac_{C_A}$ implies that $Q_{CB} = Restr_t(Q_{CA})$ is a solution for 
		$q_B \vDash ac_{C_B}$ by induction hypothesis and hence the restriction $Q_B = (q_B, Q_{CB})$ of $Q_A$ is a solution for $p_B \vDash ac_{P_B}$.
		
		\begin{figure}
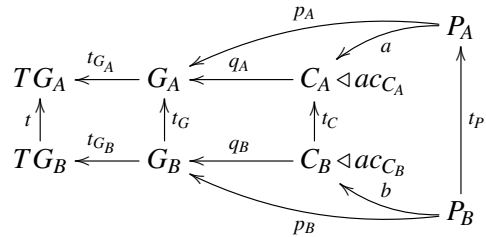

		\centering
		\xcmatrix{@R-4ex}{
			& & & P_A \ar@/_1.5ex/[dl]^{a} \ar@/_2ex/[dll]_{p_A}\\
			TG_A 
			& G_A \ar[l]_{t_{G_A}} 
			& **[r] C_A \triangleleft ac_{C_A} \ar[l]_{q_A}
			& \\
			& & & \\
			TG_B \ar[uu]^{t} 
			& G_B \ar[uu]_{t_G} \ar[l]_{t_{G_B}} 
			& **[r] C_B \triangleleft ac_{C_B} \ar[uu]_{t_C} \ar[l]_{q_B}
			& \\
			& & & P_B \ar[uuuu]_{t_P} \ar@/^1.5ex/[lu]_{b} \ar@/^2ex/[ull]^{p_B}
			}		
		
%			\xcmatrix{@R-1ex}{
%				& TG_A
%				&
%				& G_A \ar[ll] \\
%				P_A \ar[ur] \ar[rr]_/4ex/{ac_A} \ar@{-->}[urrr]^/2ex/{p_A}
%				& 
%				& C_A \ACright{ac_{C_A}} \ar[ul] \ar@{-->}[ur]_{q_A}
%				& \\
%				& TG_B \ar[uu]^(.3){t}
%				&
%				& G_B \ar[uu]_{t_G} \ar[ll] \\
%				P_B \ar[uu]^{t_P} \ar[ur] \ar[rr]_/4ex/{ac_B} \ar@{-->}[urrr]^/2ex/{p_B}
%				&
%				& C_B \ACright{ac_{C_B}} \ar@{-->}[ur]_{q_B} \ar[uu]_(.7){t_C} \ar[ul]
%				& 
%			}
			\caption{Restriction of solution $q_A$ for $p_A \vDash \exists(a, ac_{C_A})$}
			\label{fig:restriction-for-satisfaction}
		\end{figure}

	\item
		Now, for $ac_{P_A} = \bigwedge_{i \in \mathcal{I}} ac_{P_A,i}$ we have $ac_{P_B} = \bigwedge_{i \in \mathcal{I}} Restr_t(ac_{P_A,i})$. By
		the fact that $Q_A$ is a solution for $p_A \vDash ac_{P_A}$, we have that $Q_A = (Q_{A,i})_{i \in \mathcal{I}}$ such that
		$Q_{A,i}$ is a solution for $p_A \vDash ac_{A,i}$ for all $i \in \mathcal{I}$. Thus, by induction hypothesis,
		we have restrictions $Q_{B,i} = Restr_t(Q_{A,i})$ that are solutions for $p_B \vDash Restr_t(ac_{P_A,i})$
		for all $i \in \mathcal{I}$. Hence, the restriction $Q_B = (Q_{B,i})_{i \in \mathcal{I}}$ of $Q_A$ is a solution for $p_B \vDash ac_{P_B}$.

	\item
		Finally, for $ac_{P_A} = \bigvee_{i \in \mathcal{I}} ac_{P_A,i}$ we have $ac_{P_B} = \bigvee_{i \in \mathcal{I}} Restr_t(ac_{P_A,i})$.
		By the fact that $Q_A$ is a solution for $p_A \vDash ac_{P_A}$ we have that $Q_A = (Q_{A,i})_{i \in \mathcal{I}}$ 
		such that for one %some 
		$j \in \mathcal{I}$ there is a solution $Q_{A,j}$ for $p_A \vDash ac_{A,j}$ and for all $k \neq j$ we have that %there is 
		$Q_{A,k} = \emptyset$. 
		Thus, by induction hypothesis, the restriction $Q_{B,j}$ of $Q_{A,j}$ is a solution for
		$p_B \vDash Restr_t(ac_{P_A,j})$. Hence, we also have that the restriction $Q_B = (Q_{B,i})_{i \in \mathcal{I}}$ 
		is a solution for $p_B \vDash ac_{P_B}$ with $Q_{B,k} = \emptyset$ for $k \neq j$.
\end{itemize}
\end{proof}

%\begin{fact}[Amalgamation of Positive Nested Conditions]\label{fact:amalgamation-conditionsAppendix}
\newtheorem*{factAmalgamationConditions}{Fact~\ref*{fact:amalgamation-conditions}}
\begin{factAmalgamationConditions}[Amalgamation of Positive Nested Conditions]
		Given a pushout (1) as in \autoref{def:agreement-amalgamation-conditions} with all morphisms in \M. 
	\begin{description}
		\item \textbf{Composition.} 
			If there are positive nested conditions $ac_{P_B}$ and $ac_{P_C}$ typed over $TG_B$ and $TG_C$, respectively,
			agreeing in $ac_{P_D}$ typed over $TG_D$ then there exists a unique positive nested condition $ac_{P_A}$ typed over $TG_A$ such that 
			$ac_{P_A} = ac_{P_B} +_{ac_{P_D}} ac_{P_C}$.
		\item \textbf{Decomposition.}
			Vice versa, given a positive nested condition $ac_{P_A}$ typed over $TG_A$, there are unique restrictions $ac_{P_B}$, $ac_{P_C}$ and $ac_{P_D}$ 
			of $ac_{P_A}$ such that $ac_{P_A} = ac_{P_B} +_{ac_{P_D}} ac_{P_C}$.
	\end{description}
	The amalgamated composition and decomposition constructions are unique up to isomorphism.
\end{factAmalgamationConditions}

\begin{proof}\leavevmode
	\begin{description}
		\item \textbf{Composition.} We perform an induction over the structure of $ac_{P_D}$:
			\begin{itemize}
				\item $ac_{P_D} = true$. \\
					Then we also have $ac_{P_B} = true$ and $ac_{P_C} = true$, and the amalgamation $ac_{P_A}$ is trivially given by 
					$ac_{P_A} = true$.
				\item $ac_{P_D} = \exists(d, ac_{C_D})$ with $d: P_D \to C_D$. \\
					The assumption that $ac_{P_B}$ and $ac_{P_C}$ agree in $ac_{P_D}$ means that $ac_{P_D}$ is a restriction of 
					$ac_{P_B}$ and $ac_{P_C}$ and thus, by \autoref{def:restrAC}, 
					we have that $ac_{P_B} = \exists(b, ac_{C_B})$ with $b: P_B \to C_B$, $ac_{P_C} = \exists(c, ac_{C_C})$ with $c: P_C \to C_C$,
					$d$ is a restriction of $b$ and $c$, and $ac_{C_D}$ is a restriction of $ac_{C_B}$ and $ac_{C_C}$.
					This in turn means that $ac_{C_B}$ and $ac_{C_C}$ agree in $ac_{C_D}$ according to
					\autoref{def:agreement-amalgamation-conditions}. 
					So, by induction hypothesis, we obtain an amalgamation $ac_{C_A} = ac_{C_B} +_{ac_{C_D}} ac_{C_C}$, which implies that 
					$t_{CA} = t_{CB} +_{t_{CD}} t_{CC}$, i.\,e.,\ diagrams (2)-(5) below are pullbacks. By closure of \M under pullbacks, 
					we obtain from $tg_{BA}, tg_{CA} \in \M$ that also $c_{BA}, c_{CA} \in \M$.
					
					Moreover, the fact that $d$ is a restriction of $b$ and $c$ 
					means that (6)+(2) and (7)+(3) are pullbacks, which by pullback decomposition implies that (6) and (7) are pullbacks. 
					Note that $b$, $c$ and $d$ can be considered as typed over $C_B$, $C_C$ and $C_D$, respectively. 
					So, according to \autoref{def:agreement}, we obtain that $b$ and $c$ agree in $d$ with respect to the pushout of the $C$s, 
					leading to an amalgamation $a = b +_d c: P_A \to C_A$ 
					with pullbacks (8) and (9) by \autoref{fact:amalgamation}. 
					Hence, $ac_{P_A} = \exists(a, ac_{C_A})$ is the required amalgamation.
					
					\xcmatrix{@R-1.9ex@C+2ex}{
						& P_A \ACleft{ac_{P_A}} \ar[dr]_{a} \ar@{}[drrrrr]|{(8)} \ar@{}[rddddd]|{(9)}
						&&&& 
						& P_C \ACright{ac_{P_C}} \ar[lllll]_{p_{CA}} \ar[dl]^{c}
						& \\
						&& C_A \ACup{ac_{C_A}} \ar[dr]_{t_{CA}} \ar@{}[drrr]|{(4)} \ar@{}[rddd]|{(5)}
						&&
						& C_C \ACup{ac_{C_C}} \ar[lll]_{c_{CA}} \ar[dl]^{t_{CC}}
						&& \\
						&&& TG_A \ar@{}[dr]|{(1)}
						& TG_C \ar[l]_{tg_{CA}}
						&&& \\
						&&& TG_B \ar[u]^{tg_{BA}}
						& TG_D \ar[l]_{tg_{DB}} \ar[u]_{tg_{DC}}
						&&& \\
						&& C_B \ACdown{ac_{C_B}} \ar[uuu]^{c_{BA}} \ar[ur]^{t_{CB}}
						&&
						& C_D \ACdown{ac_{C_D}} \ar[lll]^{c_{DB}} \ar[uuu]_{c_{DC}} \ar[ul]_{t_{CD}} \ar@{}[ulll]|{(2)} \ar@{}[uuul]|{(3)}
						&& \\
						& P_B \ACleft{ac_{P_B}} \ar[uuuuu]^{p_{BA}} \ar[ur]^{b}
						&&&& 
						& P_D \ACright{ac_{P_D}} \ar[lllll]^{p_{DB}} \ar[uuuuu]_{p_{DC}} \ar[ul]_{d} \ar@{}[ulllll]|{(6)} \ar@{}[uuuuul]|{(7)}
						& 
					}
					
%					\item $ac_{P_D} = \neg ac_{P_D}^\prime$. \\
%						Then we also have $ac_{P_B} = \neg ac_{P_B}^\prime$ and $ac_{P_C} = \neg ac_{P_C}$. Since $ac_{P_B}$ and $ac_{P_C}$ agree in
%						$ac_{P_D}$ we obtain that also $ac_{P_B}^\prime$ and $ac_{P_C}^\prime$ agree in $ac_{P_D}^\prime$ which by induction hypothesis
%						implies $ac_{P_A}^\prime = ac_{P_B}^\prime +_{ac_{P_D}^\prime} ac_{P_C}^\prime$
%						and $t_{PA} = t_{PB} +_{t_{PD}} t_{PC}$, and hence $ac_{P_A} = \neg ac_{P_A}^\prime$
%						is the required amalgamation.
				\item $ac_{P_D} = \bigwedge_{i \in \mathcal{I}} ac_{P_D,i}$. \\
					Since $ac_{P_D}$ is a restriction of $ac_{P_B}$ and $ac_{P_C}$, they must be of the form
					$ac_{P_B} = \bigwedge_{i \in \mathcal{I}} ac_{P_B,i}$ and $ac_{P_C} = \bigwedge_{i \in \mathcal{I}} ac_{P_C,i}$.
					Moreover, since $ac_{P_B}$ and $ac_{P_C}$ agree in $ac_{P_D}$, we obtain that also $ac_{P_B,i}$ and $ac_{P_C,i}$ agree in $ac_{P_D,i}$
					for all $i \in \mathcal{I}$. So, by induction hypothesis, there are amalgamations
					$ac_{P_A,i} = ac_{P_B,i} +_{ac_{P_D,i}} ac_{P_C,i}$ such that $ac_{P_B,i}$ and $ac_{P_C,i}$ are restrictions of $ac_{P_A,i}$ 
					for all $i \in \mathcal{I}$. Hence, $ac_{P_A} = \bigwedge_{i \in \mathcal{I}} ac_{P_A,i}$ is the required amalgamation.
				\item The remaining case for disjunction works analogously to the case for conjunction.
			\end{itemize}
			The uniqueness of the amalgamation follows from the fact that we have an amalgamation in each level of nesting and the amalgamation
			of typed objects is unique by \autoref{fact:amalgamation}.
			
		\item \textbf{Decomposition.} We do an induction over the structure of $ac_{P_A}$:
			\begin{itemize}
				\item $ac_{P_A} = true$. \\
					This case is trivial because $true = true +_{true} true$.
				\item $ac_{P_A} = \exists(a, ac_{C_A})$ with $a: P_A \rightarrow C_A$. \\
					Then by induction hypothesis, there exist restrictions $ac_{C_B}$, $ac_{C_C}$ and $ac_{C_D}$ of $ac_{C_A}$ such that 
					$ac_{C_A} = ac_{C_B} +_{ac_{C_D}} ac_{C_C}$. Moreover, by \autoref{fact:amalgamation}, there are unique restrictions
					$b$, $c$ and $d$ of $a$ such that $a = b +_d c$. Hence, we have restrictions $ac_{P_B} = \exists(b, ac_{C_B})$,
					$ac_{P_C} = \exists(c, ac_{C_C})$ and $ac_{P_D} = \exists(d, ac_{C_D})$ of $ac_{P_A}$, and, as shown for the case of composition before, the fact that
					$ac_{C_A} = ac_{C_B} +_{ac_{C_D}} ac_{C_C}$ and $a = b +_d c$ implies that $ac_{P_A} = ac_{P_B} +_{ac_{P_D}} ac_{P_C}$.
%					\item $ac_{P_A} = \neg ac_{P_A}^\prime$. \\
%						Then by induction hypothesis there exist restrictions $ac_{P_B}^\prime$, $ac_{P_C}^\prime$ and $ac_{P_D}^\prime$
%						of $ac_{P_A}^\prime$ such that $ac_{P_A}^\prime = ac_{P_B}^\prime +_{ac_{P_D}^\prime} ac_{P_C}^\prime$. Hence, 
%						$ac_{P_B} = \neg ac_{P_B}\prime$, $ac_{P_C} = \neg ac_{P_C}\prime$ and $ac_{P_D} = \neg ac_{P_D}\prime$ are restrictions
%						of $ac_{P_A}$ such that $ac_{P_A} = ac_{P_B} +_{ac_{P_D}} ac_{P_C}$.
				\item $ac_{P_A} = \bigwedge_{i \in \mathcal{I}} ac_{P_A,i}$. \\
					Then by induction hypothesis, there exist restrictions $ac_{P_B,i}$, $ac_{P_C,i}$ and $ac_{P_D,i}$ of $ac_{P_A,i}$
					such that $ac_{P_A,i} = ac_{P_B,i} +_{ac_{P_D,i}} ac_{P_C,i}$ for all $i \in \mathcal{I}$. Hence, 
					$ac_{P_B} = \bigwedge_{i \in \mathcal{I}} ac_{P_B,i}$, $ac_{P_C} = \bigwedge_{i \in \mathcal{I}} ac_{P_C,i}$
					and $ac_{P_D} = \bigwedge_{i \in \mathcal{I}} ac_{P_D,i}$ are restrictions of $ac_{P_A}$ such that
					$ac_{P_A} = ac_{P_B} +_{ac_{P_D}} ac_{P_C}$.
				\item Again, the remaining case for disjunction works analogously to the case for conjunction.
			\end{itemize}
			The uniqueness of the decomposition follows from the uniqueness of restrictions by pullback construction.
	\end{description}
\end{proof}

%\begin{theorem}[Amalgamation of Solutions for Positive Nested Conditions]\label{thm:ACviaRestrAppendix}
\newtheorem*{thmACviaRestr}{Theorem~\ref*{thm:ACviaRestr}}
\begin{thmACviaRestr}[Amalgamation of Solutions for Positive Nested Conditions]
Given pushout (1) as in \autoref{def:agreement-amalgamation-solution} with all morphisms in \M, 
an amalgamation of typed objects $g_A = g_B +_{g_D} g_C$,
and an amalgamation of positive nested conditions $ac_{P_A} = ac_{P_B} +_{ac_{P_D}} ac_{P_C}$ 
with corresponding matches $p_A = p_B +_{p_D} p_C$.
\begin{description}
	\item \textbf{Composition.}
		Given solutions $Q_B$ for $p_B \vDash ac_{P_B}$ and $Q_C$ for $p_C \vDash ac_{P_C}$ agreeing in a solution $Q_D$ for $p_D \vDash ac_{P_D}$,
		then there is a solution $Q_A$ for $p_A \vDash ac_{P_A}$ constructed as 
		amalgamation $Q_A = Q_B +_{Q_D} Q_C$.
	\item \textbf{Decomposition.}
		Given a solution $Q_A$ for $p_A \vDash ac_{P_A}$, then there are 
		solutions $Q_B$, $Q_C$ and $Q_D$ for 
		$p_B \vDash ac_{P_B}$, $p_C \vDash ac_{P_C}$ 
		and $p_D \vDash ac_{P_D}$, respectively, which are constructed as
		restrictions $Q_B$, $Q_C$ and $Q_D$ of $Q_A$ such that
		$Q_A = Q_B +_{Q_D} Q_C$.
\end{description}
The amalgamated composition and decomposition constructions are unique up to isomorphism.
\end{thmACviaRestr}

\begin{proof}\leavevmode
	\begin{description}
		\item \textbf{Composition.}
			We perform an induction over the structure of $ac_{P_A}$.
			\begin{itemize}
			\item $ac_{P_A} = true$. \\
				Then also $ac_{P_B}$, $ac_{P_C}$, $ac_{P_D}$ are true and we have empty solutions $Q_A$, $Q_B$, $Q_C$ and $Q_D$. Since the restriction of
				an empty solution is empty, we have that $Q_B$ and $Q_C$ are restrictions of $Q_A$.
				
			\item $ac_{P_A} = \exists(a, ac_{C_A})$ with $a: P_A \to C_A$. \\
%				Then we have $ac_{P_B} = \exists(b, ac_{C_B})$, $ac_C = \exists(c, ac_{C_C})$ and $ac_D = \exists(d, ac_{C_D})$ such that
%				$b$, $c$ and $d$ are restrictions of $a$ which according to Definition~\ref{def:restr} means that the diagrams (2)-(5) and (2')-(5')
%				below are pullbacks. 
%				By \M-morphisms $tg_{BA}$, $tg_{CA}$, $tg_{DC}$ and $tg_{DB}$ and closure of \M-morphisms under pullbacks we obtain that
%				the corresponding morphisms in the outer diagrams are in \M. Thus, using the horizontal VK-property, by pushout of the $TG$s and
%				pullbacks (2)-(5) we obtain that the diagram of the $C$s is a pushout.
				By \autoref{fact:amalgamation-conditions} (Composition), we have the following diagram, where all rectangles are pushouts
				and all trapezoids are pullbacks, and all horizontal and vertical morphisms are in \M.
				
				\xcmatrix{@R-2.1ex@C+2.2ex}{
					& P_A \ACleft{ac_{P_A}} \ar[dr]_{a} \ar@{}[drrrrr]|{(2')} \ar@{}[rddddd]|{(3')}
					&&&& 
					& P_C \ACright{ac_{P_C}} \ar[lllll]_{p_{CA}} \ar[dl]^{c}
					& \\
					&& C_A \ACup{ac_{C_A}} \ar[dr]_{t_{CA}} \ar@{}[drrr]|{(2)} \ar@{}[rddd]|{(3)}
					&&
					& C_C \ACup{ac_{C_C}} \ar[lll]_{c_{CA}} \ar[dl]^{t_{CC}}
					&& \\
					&&& TG_A \ar@{}[dr]|{(1)}
					& TG_C \ar[l]_{tg_{CA}} 
					&&& \\
					&&& TG_B \ar[u]^{tg_{BA}}
					& TG_D \ar[l]_{tg_{DB}} \ar[u]_{tg_{DC}}
					&&& \\
					&& C_B \ACdown{ac_{C_B}} \ar[uuu]^{c_{BA}} \ar[ur]^{t_{CB}}
					&&
					& C_D \ACdown{ac_{C_D}} \ar[lll]^{c_{DB}} \ar[uuu]_{c_{DC}} \ar[ul]_{t_{CD}} \ar@{}[ulll]|{(4)} \ar@{}[uuul]|{(5)}
					&& \\
					& P_B \ACleft{ac_{P_B}} \ar[uuuuu]^{p_{BA}} \ar[ur]^{b}
					&&&& 
					& P_D \ACright{ac_{P_D}} \ar[lllll]^{p_{DB}} \ar[uuuuu]_{p_{DC}} \ar[ul]_{d} \ar@{}[ulllll]|{(4')} \ar@{}[luuuuu]|{(5')}
					& 
				}

				Now, we consider solutions $Q_B = (q_B, Q_{CB})$, $Q_C = (q_C, Q_{CC})$ and $Q_D = (q_D, Q_{CD})$ for $p_B \vDash ac_{P_B}$, 
				$p_C \vDash ac_{P_C}$ and $p_D \vDash ac_{P_D}$, respectively, 
				such that $Q_D$ is a restriction of $Q_B$ and $Q_C$. Then we also have that $q_D$ is a restriction
				of $q_B$ and $q_C$, and thus
				\[ g_{BA} \circ q_B \circ c_{DB} = g_{BA} \circ g_{DB} \circ q_D = g_{CA} \circ g_{DC} \circ q_D = g_{CA} \circ q_C \circ c_{DC}. \]
				Together with the pushout over the $C$s, this implies a unique morphism $q_A: C_A \to G_A$ with $q_A \circ c_{BA} = g_{BA} \circ q_B$ and 
				$q_A \circ c_{CA} = g_{CA} \circ q_C$.
		
				\xcmatrix{@R-3.6ex}{
					P_A \ACleft{ac_{P_A}} \ar[dr]_(.3){a} \ar@/^1.5ex/[]+<1.5ex,-0.5ex>;[drdr]+<0ex,1.5ex>^(.5){p_A}
					&&&&&& 
					& P_C \ACright{ac_{P_C}} \ar[lllllll]_{p_{CA}} \ar[dl]^(.3){c} \ar@/_1.5ex/[]+<-1.5ex,-0.5ex>;[dldl]+<0ex,1.5ex>_(.5){p_C}
					& \\
					& C_A \ACleft[-4.5ex,-0.5ex]{ac_{C_A}} \ar[dr]_{q_A}
					&&&& 
					& C_C \ACright[4.5ex, -0.5ex]{ac_{C_C}} \ar[lllll]^{c_{CA}}|(.1)*+{\hole}|(.9)*+{\hole} \ar[dl]^{q_C}
					& \\
					&& G_A \ar[dr]_{g_A}
					&&
					& G_C \ar[lll]^{g_{CA}} \ar[dl]^{g_C}
					&& \\
					&&& TG_A
					& TG_C \ar[l]_{tg_{CA}}
					&&& \\
					&&& TG_B \ar[u]^{tg_{BA}}
					& TG_D \ar[l]_{tg_{DB}} \ar[u]_{tg_{DC}}
					&&& \\
					&& G_B \ar[uuu]_{g_{BA}} \ar[ur]^{g_B}
					&&
					& G_D \ar[lll]^{g_{DB}} \ar[uuu]_{g_{DC}} \ar[ul]_{g_D}
					&& \\
					& C_B \ACleft[-4.5ex,0.4ex]{ac_{C_B}} \ar[uuuuu]_{c_{BA}} \ar[ur]^{q_B}
					&&&& 
					& C_D \ACright[4.5ex,0.4ex]{ac_{C_D}} \ar[lllll]^{c_{DB}}|(.1)*+{\hole}|(.9)*+{\hole} \ar[uuuuu]_{c_{DC}} \ar[ul]_{q_D}
					& \\
					P_B \ACleft{ac_{P_B}} \ar[uuuuuuu]^{p_{BA}} \ar[ur]^(.3){b} \ar@/_1.5ex/[]+<1.5ex,0.5ex>;[urur]+<0ex,-1.5ex>_(.5){p_B}
					&&&&&& 
					& P_D \ACright{ac_{P_D}} \ar[lllllll]^{p_{DB}} \ar[uuuuuuu]_{p_{DC}} \ar[ul]_(.3){d}
						\ar@/^1.5ex/[]+<-1.5ex,0.5ex>;[ulul]+<0ex,-1.5ex>^(.5){p_D}
					& 
				}
				
				Moreover, we have
				\[ q_A \circ a \circ p_{BA} = q_A \circ c_{BA} \circ b = g_{BA} \circ q_B  \circ b = g_{BA} \circ p_B = p_A \circ p_{BA} \]
				and analogously $q_A \circ a \circ p_{CA} = p_A \circ p_{CA}$. Since $p_{BA}$ and $p_{CA}$ are jointly epimorphic, this implies that $q_A \circ a = p_A$.
				
				In order to show that $q_A \in \M$, we consider the following diagram in the left:
				
				\begin{center}
					\begin{minipage}{5.5cm}
						\xcmatrix{@R+1ex@C+2ex}{
							G_A & G_C \ar[l]_{g_{CA}} & C_C \ar[l]_{q_C}  \\
							G_B \ar[u]^{g_{BA}} & G_D \ar[l]^{g_{DB}} \ar[u]^{g_{DC}} \ar@{}[dr]|{(9)} \ar@{}[dl]|{(7)} \ar@{}[ur]|{(8)} \ar@{}[ul]|{(6)}
								& C_D \ar[l]^{q_D} \ar[u]_{c_{DC}} \\
							C_B \ar[u]^{q_B} & C_D \ar[l]^{c_{DB}} \ar[u]^{q_D} & C_D \ar[l]^{id_{C_D}} \ar[u]_{id_{C_D}} 
						}
					\end{minipage}
					\qquad
					\begin{minipage}{5.5cm}
						\xcmatrix{@R+0.2ex@C+0.5ex}{
						& & C_C \ar[dll]_{g_{CA} \circ q_C} \ar[dl]^{c_{CA}} & \\
						G_A & C_A \ar[l]|{q_A} & & C_D \ar[ul]_{c_{DC}} \ar[dl]^{c_{DB}} \\
						& & C_B \ar[ull]^{g_{BA} \circ q_B} \ar[ul]_{c_{BA}}
						}
					\end{minipage}
				\end{center}
				
				We have that (6) is a pushout with all morphisms in \M and thus also a pullback. Diagrams (7) and (8) are pullbacks by restriction,
				and (9) is a pullback because $q_D \in \M$ is a monomorphism. Hence, by composition of pullbacks, we obtain that the complete diagram
				is a pullback along \M-morphisms $g_{BA} \circ q_B$ and $g_{CA} \circ q_C$, which means that the pushout of the $C$s is effective
				(see \autoref{def:EfPO}), implying that $q_A \in \M$.
				
				It remains to show that $q_B$ and $q_C$ are restrictions of $q_A$. 
				In the following diagram, we have that (10) and (11) are pullbacks by restrictions, the $C$s and the $G$s form pushouts 
				(see \autoref{rem:agreement-amalgamation-solution})
				and all morphisms in (10)-(13) are in \M. 
				So, the horizontal as well as the vertical VK property implies that also (12) and (13) are pullbacks, which means that $q_B$ and $q_C$ 
				are restrictions of $q_A$.
		
				\xcmatrix{@R-3ex@C+2ex}{
					& C_A \ACleft{ac_{C_A}} \ar[dr]_{q_A} \ar@{}[drrrrr]|{(12)} \ar@{}[rddddd]|{(13)}
					&&&& 
					& C_C \ACright{ac_{C_C}} \ar[lllll]^{c_{CA}} \ar[dl]^{q_C} 
					& \\
					&& G_A \ar[dr]_{g_A}
					&&
					& G_C \ar[lll]^{g_{CA}} \ar[dl]^{g_C}
					&& \\
					&&& TG_A
					& TG_C \ar[l]_{tg_{CA}}
					&&& \\
					&&& TG_B \ar[u]^{tg_{BA}}
					& TG_D \ar[l]_{tg_{DB}} \ar[u]_{tg_{DC}}
					&&& \\
					&& G_B \ar[uuu]_{g_{BA}} \ar[ur]^{g_B}
					&&
					& G_D \ar[lll]^{g_{DB}} \ar[uuu]_{g_{DC}} \ar[ul]_{g_D}
					&& \\
					& C_B \ACleft{ac_{C_B}} \ar[uuuuu]_{c_{BA}} \ar[ur]^{q_B}
					&&&& 
					& C_D \ACright{ac_{C_D}} \ar[lllll]^{c_{DB}} \ar[uuuuu]_{c_{DC}} \ar[ul]_{q_D} \ar@{}[ulllll]|{(10)} \ar@{}[luuuuu]|{(11)}
					& 
				}
				
				Finally, $Q_D$ being a restriction of $Q_B$ and $Q_C$ means that $Q_{CD}$ is a restriction of $Q_{CB}$ and $Q_{CC}$
				by induction hypothesis, this implies a solution $Q_{CA}$ of $q_A \vDash ac_{C_A}$ such that $Q_{CB}$ and $Q_{CC}$ are restrictions of $Q_{CA}$.
				Hence, $Q_A = (q_A, Q_{CA})$ is a solution for $p_A \vDash ac_A$ such that $Q_B$ and $Q_C$ are restrictions of $Q_A$.
				
			\item $ac_{P_A} = \bigwedge_{i \in \mathcal{I}} ac_{P_A,i}$. \\
				We have $ac_{P_B} = \bigwedge_{i \in \mathcal{I}} ac_{P_B,i}$, $ac_{P_C} = \bigwedge_{i \in \mathcal{I}} ac_{P_C,i}$ and
				$ac_{P_D} = \bigwedge_{i \in \mathcal{I}} ac_{P_D,i}$ such that for all $i \in \mathcal{I}$ there is $ac_{P_D,i}$ a restriction of 
				$ac_{P_B,i}$ and $ac_{P_C,i}$.
				
				Moreover, given solutions $Q_B$, $Q_C$ and $Q_D$ of $p_B \vDash ac_{P_B}$, $p_C \vDash ac_{P_C}$ and $p_D \vDash ac_{P_D}$, respectively,
				we have $Q_B = (Q_{B,i})_{i \in \mathcal{I}}$, $Q_C = (Q_{C,i})_{i \in \mathcal{I}}$ and $Q_D = (Q_{D,i})_{i \in \mathcal{I}}$
				such that for all $i \in \mathcal{I}$ we have that $Q_{B,i}$, $Q_{C,i}$ and $Q_{D,i}$ are solutions for $p_B \vDash ac_{P_B,i}$, 
				$p_C \vDash ac_{P_C,i}$ and $p_D \vDash ac_{P_D,i}$, respectively, and $Q_{D,i}$ is a restriction of $Q_{B,i}$ and $Q_{C,i}$.
				
				Then, by induction hypothesis, there are solutions $Q_{A,i}$ for $p_A \vDash ac_{P_A,i}$ for all $i \in \mathcal{I}$ such that $Q_{B,i}$
				and $Q_{C,i}$ are restrictions of $Q_{A,i}$. Hence, $Q_A = (Q_{A,i})_{i \in \mathcal{I}}$ is the required solution for $p_A \vDash ac_{P_A}$.
				
			\item $ac_{P_A} = \bigvee_{i \in \mathcal{I}} ac_{P_A,i}$. \\
				We have $ac_{P_B} = \bigvee_{i \in \mathcal{I}} ac_{P_B,i}$, $ac_{P_C} = \bigvee_{i \in \mathcal{I}} ac_{P_C,i}$ and
				$ac_{P_D} = \bigvee_{i \in \mathcal{I}} ac_{P_D,i}$ such that for all $i \in \mathcal{I}$ there is $ac_{P_D,i}$ a restriction of $ac_{P_B,i}$
				and $ac_{P_C,i}$.
				
				Moreover, given solutions $Q_B$, $Q_C$ and $Q_D$ of $p_B \vDash ac_{P_B}$, $p_C \vDash ac_{P_C}$ and $p_D \vDash ac_{P_D}$, respectively.
				Then we have $Q_B = (Q_{B,i})_{i \in \mathcal{I}}$, $Q_C = (Q_{C,i})_{i \in \mathcal{I}}$ and $Q_D = (Q_{D,i})_{i \in \mathcal{I}}$
				such that for some $j_B,j_C,j_D \in \mathcal{I}$ we have that $Q_{B,j_B}$, $Q_{C,j_C}$ and $Q_{D,j_D}$ are solutions for $p_B \vDash ac_{P_B,j_B}$, 
				$p_C \vDash ac_{P_C,j_C}$ and $p_D \vDash ac_{P_D,j_D}$, respectively, and for all $k_B, k_C, k_D \in \mathcal{I}$ with $k_B \neq j_B$,
				$k_C \neq j_C$ and $k_D \neq j_D$ we have that $Q_{B,k_B}$, $Q_{C, k_C}$ and $Q_{D,k_D}$ are empty.
				Furthermore, $Q_{D,i}$ is a restriction of $Q_{B,i}$ and $Q_{C,i}$ for all $i \in \mathcal{I}$ .
				
				\begin{description}
					\item \textbf{Case 1.} $Q_{D,j_D} = \emptyset$. \\
						Then we have $Q_{D,j} = \emptyset$ for all $j \in \mathcal{I}$. 
						According to \autoref{def:restrSolution}, only the restriction of an empty solution is empty, implying that we also
						have $Q_{B,j} = Q_{C,j} = \emptyset$ for all $j \in \mathcal{I}$.
						Moreover, since $Q_{D,j_D}$ is a solution for $p_D \vDash ac_{P_D,j_D}$, we can conclude that 
						$ac_{P_D,j_D} = true$,
						and by the fact that $ac_{P_D,j_D}$ is a restriction of $ac_{P_A,j_D}$, $ac_{P_B,j_D}$ and $ac_{P_C,j_D}$
						it follows that also $ac_{P_A,j_D} = true$, $ac_{P_B,j_D} = true$ and $ac_{P_C,j_D} = true$. 
						So, as shown above, there is a solution $Q_{A,j_D} = \emptyset$ for $p_A \vDash ac_{P_A,j_D}$. Hence, $Q_A = (Q_{A,i})_{i \in \mathcal{I}}$
						with $Q_{A,i} = \emptyset$ for all $i \in \mathcal{I}$ is a solution for $p_A \vDash ac_{P_A}$ such that $Q_B$ and $Q_C$ are restrictions
						of $Q_A$.
					\item \textbf{Case 2.} $Q_{D,j_D} \neq \emptyset$. \\
						Then according to \autoref{def:restrSolution}, there are also $Q_{B,j_D} \neq \emptyset$ and $Q_{C,j_D} \neq \emptyset$
						which means that $j_B = j_C = j_D$. So, by induction hypothesis, there is a solution $Q_{A,j_D}$ for $p_A \vDash ac_{P_A,j_D}$
						such that $Q_{B,j_D}$ and $Q_{C,j_D}$ are restrictions of $Q_{A,j_D}$. Hence, $Q_A = (Q_{A,i})_{i \in \mathcal{I}}$ with
						$Q_{A,k} = \emptyset$ for all $k \in \mathcal{I}$ with $k \neq j_D$ is a solution for $p_A \vDash ac_{P_A}$, and we have that
						$Q_B$ and $Q_C$ are restrictions of $Q_A$.
				\end{description}
			\end{itemize}
			In the first case ($ac_{P_A} = true$), the uniqueness of the amalgamation follows from the fact that an empty solution can only be the restriction
			of another empty solution. In the second case ($ac_{P_A} = \exists(a, ac_{C_A})$), the uniqueness of $Q_A = (q_A, Q_{CA})$ 
			follows from the uniqueness of $q_A$ by universal pushout property, and by uniqueness of $Q_{CA}$ by induction hypothesis. Finally,
			in the cases of conjunction and disjunction, the uniqueness of the solution follows from uniqueness of its components by induction hypothesis.
		\item \textbf{Decomposition.} Again, we perform an induction over the structure of $ac_{P_A}$.
			\begin{itemize}
				\item $ac_{P_A} = true$. \\
					Then we also have that $ac_{P_B}$, $ac_{P_C}$ and $ac_{P_D}$ are true. Moreover, we have that $Q_A$ is empty, leading to empty
					restrictions $Q_B$, $Q_C$ and $Q_D$ that are solutions for $p_B \vDash ac_{P_B}$, $p_C \vDash ac_{P_C}$ and $p_D \vDash ac_{P_D}$,
					respectively.
				\item $ac_{P_A} = \exists(a, ac_{C_A})$ with $a: P_A \rightarrow C_A$. \\
					Then we have $ac_{P_B} = \exists(b, ac_{C_B})$, $ac_{P_C} = \exists(c, ac_{C_C})$ and $ac_{P_D} = \exists(d, ac_{C_D})$.
					By amalgamation $g_A = g_B +_{g_D} g_C$, we have pullbacks (2)-(5) below. Moreover, by restrictions $ac_{P_B}$, $ac_{P_C}$ and
					$ac_{P_D}$ of $ac_{P_A}$, we have restrictions $b$, $c$ and $d$ of $a$, implying pullbacks (6)-(9) below. 
					According to \autoref{rem:amalgamation-conditions}, we have an amalgamation of positive nested conditions
					$ac_{C_A} = ac_{C_B} +_{ac_{C_D}} ac_{C_C}$, which implies an amalgamation of typed objects $t_{CA} = t_{CB} +_{t_{CD}} t_{CC}$
					by \autoref{def:agreement-amalgamation-conditions}.
					
					\xcmatrix{@R-2.8ex@C+0.2ex}{
						P_A \ACleft{ac_{P_A}} \ar[dr]_(.3){a} \ar@/^1.5ex/[]+<1.5ex,-0.5ex>;[drdr]+<0ex,1.5ex>^(.5){p_A}
							\ar@{}[dddddddr]|{(6)} \ar@{}[rrrrrrrd]|{(7)}
						&&&&&& 
						& P_C \ACright{ac_{P_C}} \ar[lllllll]_{p_{CA}} \ar[dl]^(.3){c} \ar@/_1.5ex/[]+<-1.5ex,-0.5ex>;[dldl]+<0ex,1.5ex>_(.5){p_C}
						& \\
						& C_A \ACleft[-4.5ex,-0.5ex]{ac_{C_A}} \ar@/_2.5ex/[]+<0.5ex,-1.5ex>;[drdr]+<-2.5ex,-0.3ex>_{t_{CA}} \ar[dr]_(.7){q_A}
							\ar@{}[dddddr]|{(10)} \ar@{}[rrrrrd]|{(11)}
						&&&& 
						& C_C \ACright[4.5ex, -0.5ex]{ac_{C_C}} \ar@/^2.5ex/[]+<-0.5ex,-1.5ex>;[dldl]+<2.5ex,-0.3ex>^{t_{CC}}
							\ar[lllll]^{c_{CA}}|(.1)*+{\hole}|(.9)*+{\hole} \ar@{.>}[dl]^(.7){q_C}
						& \\
						&& G_A \ar[dr]_{g_A} \ar@{}[dddr]|{(2)} \ar@{}[rrrd]|{(3)}
						&&
						& G_C \ar[lll]^{g_{CA}} \ar[dl]^{g_C}
						&& \\
						&&& TG_A \ar@{}[dr]|{(1)}
						& TG_C \ar[l]_{tg_{CA}}
						&&& \\
						&&& TG_B \ar[u]^{tg_{BA}}
						& TG_D \ar[l]_{tg_{DB}} \ar[u]_{tg_{DC}}
						&&& \\
						&& G_B \ar[uuu]_{g_{BA}}|(.25)\hole|(.75)\hole \ar[ur]^{g_B}
						&&
						& G_D \ar[lll]^{g_{DB}} \ar[uuu]_{g_{DC}}|(.25)\hole|(.75)\hole \ar[ul]_{g_D} \ar@{}[lllu]|{(4)} \ar@{}[uuul]|{(5)}
						&& \\
						& C_B \ACleft[-4.5ex,0.4ex]{ac_{C_B}} \ar@/^2.5ex/[]+<0.5ex,1.5ex>;[urur]+<-2.5ex,0.3ex>^{t_{CB}} \ar[uuuuu]_{c_{BA}} 
							\ar@{.>}[ur]^(.7){q_B}
						&&&& 
						& C_D \ACright[4.5ex,0.4ex]{ac_{C_D}} \ar@/_2.5ex/[]+<-0.5ex,1.5ex>;[ulul]+<2.5ex,0.3ex>_{t_{CD}}
							\ar[lllll]^{c_{DB}}|(.1)*+{\hole}|(.9)*+{\hole} \ar[uuuuu]_{c_{DC}} \ar@{.>}[ul]_(.7){q_D}
							\ar@{}[lllllu]|{(12)} \ar@{}[uuuuul]|{(13)}
						& \\
						P_B \ACleft{ac_{P_B}} \ar[uuuuuuu]^{p_{BA}} \ar[ur]^(.3){b} \ar@/_1.5ex/[]+<1.5ex,0.5ex>;[urur]+<0ex,-1.5ex>_(.5){p_B}
						&&&&&& 
						& P_D \ACright{ac_{P_D}} \ar[lllllll]^{p_{DB}} \ar[uuuuuuu]_{p_{DC}} \ar[ul]_(.3){d}
							\ar@/^1.5ex/[]+<-1.5ex,0.5ex>;[ulul]+<0ex,-1.5ex>^(.5){p_D}
							\ar@{}[lllllllu]|{(8)} \ar@{}[uuuuuuul]|{(9)}
						& 
					}
					
					Now, given a solution $Q_A = (q_A, Q_{CA})$ for $p_A \vDash ac_{P_A}$, there is $q_A: C_A \to G_A \in \M$ with $q_A \circ a = p_A$.
					
					Furthermore, we have
					\[ g_A \circ q_A \circ c_{BA} = t_{CA} \circ c_{BA} = tg_{BA} \circ t_{CB}, \]
					which implies a unique morphism $q_B: C_B \to G_B$ by pullback (2) such that 
					$g_B \circ q_B = t_{CB}$ and $g_{BA} \circ q_B = q_A \circ c_{BA}$. Due to amalgamation $t_{CA} = t_{CB} +_{t_{CD}} t_{CC}$, we have that $t_{CB}$
					is a restriction  of $t_{CA}$ and thus (10)+(2) is a pullback. So, together with pullback (2), we obtain that also (10) is
					a pullback by pullback decomposition and, thus, $q_B$ is a restriction of $q_A$. 
					
					Moreover, by $q_A, tg_{BA} \in \M$ and closure of \M under pullbacks, 
					we know that
%					there are also 
					$q_B, g_{BA} \in \M$. Hence, by 
					\[ g_{BA} \circ p_B = p_A \circ p_{BA} = q_A \circ a \circ p_{BA} = q_A \circ c_{BA} \circ b = g_{BA} \circ q_B \circ b, \]
					we obtain $p_B = q_B \circ b$ because $g_{BA} \in \M$ is a monomorphism. 
					
					Analogously, due to pullback (3) and restriction $t_{CC}$ of $t_{CA}$, there is a unique restriction $q_C: C_C \to G_C \in \M$ 
					of $q_A$ with pullback (11) such that $p_C = q_C \circ c$, 
					and due to pullback (4) and restriction $t_{CD}$ of $t_{CB}$, there is a unique 
					restriction $q_D: C_D \to G_D \in \M$ of $q_B$ with pullback (12) such that $p_D = q_D \circ d$. Then, since $t_{CD}$ is a restriction
					of $t_{CC}$, (5)+(13) is a pullback which implies that also (13) is a pullback by pullback decomposition and pullback (5). Thus, 
					$q_D$
					%there 
					is also a restriction of $q_C$, which means that we have $q_A = q_B +_{q_D} q_C$.
					
					So, by induction hypothesis, there are solutions $Q_{CB}$ for $q_B \vDash ac_{C_B}$, $Q_{CC}$ for $q_C \vDash ac_{C_C}$,
					and $Q_{CD}$ for $q_D \vDash ac_{C_D}$ such that $Q_{CA} = Q_{CB} +_{Q_{CD}} Q_{CC}$. Hence, for $Q_B = (q_B, Q_{CB})$,
					$Q_C = (q_C, Q_{CC})$ and $Q_D = (q_D, Q_{CD})$ we obtain that $Q_A = Q_B +_{Q_D} Q_C$.
					
				\item $ac_{P_A} = \bigwedge_{i \in \mathcal{I}} ac_{P_A,i}$. \\
					Then we also have $ac_{P_B} = \bigwedge_{i \in \mathcal{I}} ac_{P_B,i}$, $ac_{P_C} = \bigwedge_{i \in \mathcal{I}} ac_{P_C,i}$,
					and $ac_{P_D} = \bigwedge_{i \in \mathcal{I}} ac_{P_D,i}$.
					Now, given a solution $Q_A = (Q_{A,i})_{i \in \mathcal{I}}$ for $p_A \vDash ac_{P_A}$, 
					then $Q_{A,i}$ is a solution for $p_A \vDash ac_{P_A,i}$ for all $i \in \mathcal{I}$. Thus, by induction hypothesis for all 
					$i \in \mathcal{I}$, there are solutions $Q_{B,i}$ for $p_B \vDash ac_{P_B,i}$, $Q_{C,i}$ for $p_C \vDash ac_{P_C,i}$,
					and $Q_{D,i}$ for $ac_{P_D,i}$ such that $Q_{A,i} = Q_{B,i} +_{Q_{D,i}} Q_{C,i}$. This in turn means that for all $i \in \mathcal{I}$
					there are $Q_{B,i}$ and $Q_{C,i}$ restrictions of $Q_{A,i}$, and $Q_{D,i}$ is a restriction of $Q_{B,i}$ and $Q_{C,i}$.
					Hence, for $Q_B = (Q_{B,i})_{i \in \mathcal{I}}$,
					$Q_C = (Q_{C,i})_{i \in \mathcal{I}}$ and $Q_D = (Q_{D,i})_{i \in \mathcal{I}}$ we have that $Q_B$ and $Q_C$ are restrictions of $Q_A$,
					and $Q_D$ is a restriction of $Q_B$ and $Q_C$, implying $Q_A = Q_B +_{Q_D} Q_C$.
				\item $ac_{P_A} = \bigvee_{i \in \mathcal{I}} ac_{P_A,i}$. \\
					Then we also have $ac_{P_B} = \bigvee_{i \in \mathcal{I}} ac_{P_B,i}$, $ac_{P_C} = \bigvee_{i \in \mathcal{I}} ac_{P_C,i}$,
					and $ac_{P_D} = \bigvee_{i \in \mathcal{I}} ac_{P_D,i}$.
					Given a solution $Q_A = (Q_{A,i})_{i \in \mathcal{I}}$ for $p_A \vDash ac_{P_A}$, then there is $j \in \mathcal{I}$, 
					such that $Q_{A,j}$ is a solution for $p_A \vDash ac_{P_A,j}$, and for all $k \in \mathcal{I}$ with $k \neq j$ there is $Q_{A,k} = \emptyset$.
					By induction hypothesis, there are solutions $Q_{B,j}$ for $p_B \vDash ac_{P_B,j}$, $Q_{C,j}$ for $p_C \vDash ac_{P_C,j}$, and
					$Q_{D,j}$ for $p_D \vDash ac_{P_D,j}$ such that $Q_{A,j} = Q_{B,j} +_{Q_{D,j}} Q_{C,j}$. Hence, for 
					$Q_B = (Q_{B,i})_{i \in \mathcal{I}}$,$Q_C = (Q_{C,i})_{i \in \mathcal{I}}$ and $Q_D = (Q_{D,i})_{i \in \mathcal{I}}$, 
					where for all $k \in \mathcal{I}$ with $k \neq j$ there is $Q_{B,k} = Q_{C,k} = Q_{D,k} = \emptyset$, we have that
					$Q_A = Q_B +_{Q_D} Q_C$.
			\end{itemize}
			The uniqueness of the solutions follows from uniqueness of restrictions by pullback constructions.
	\end{description}
\end{proof}
\end{appendix}

%\newpage

\end{document}